\documentclass[english]{article}
\usepackage[latin9]{inputenc}
\usepackage{geometry}
\usepackage{float}
\usepackage{amsmath}
\usepackage{amsthm}
\usepackage{graphicx}
\usepackage{wasysym}
\usepackage{booktabs}

\makeatletter

\providecommand{\tabularnewline}{\\}
\floatstyle{ruled}
\newfloat{algorithm}{tbp}{loa}
\providecommand{\algorithmname}{Algorithm}
\floatname{algorithm}{\protect\algorithmname}

\numberwithin{equation}{section}
\numberwithin{figure}{section}
\theoremstyle{plain}
\newtheorem{thm}{\protect\theoremname}[section]
\theoremstyle{plain}
\newtheorem{conjecture}[thm]{\protect\conjecturename}
\theoremstyle{plain}
\newtheorem{cor}[thm]{\protect\corollaryname}
\theoremstyle{definition}
\newtheorem{defn}[thm]{\protect\definitionname}
\theoremstyle{plain}
\newtheorem{lem}[thm]{\protect\lemmaname}
\theoremstyle{plain}
\newtheorem{fact}[thm]{\protect\factname}
\theoremstyle{remark}
\newtheorem*{rem*}{\protect\remarkname}

\usepackage{url}
\usepackage{xargs}
\usepackage{amsfonts}
\usepackage{amssymb}
\usepackage[pdftex,dvipsnames]{xcolor}
\usepackage[colorinlistoftodos,prependcaption,textsize=tiny]{todonotes}
\usepackage{algorithm}
\usepackage{algpseudocode}

\usepackage[small,bf]{caption}
\usepackage{tikz}
\usepackage{enumitem}
\setlist{nolistsep}
\usepackage{hyperref}
\hypersetup{
    colorlinks,
    citecolor=red,
    filecolor=red,
    linkcolor=red,
    urlcolor=red
}

\usepackage{thmtools}
\usepackage{thm-restate}

\providecommand{\conjecturename}{Conjecture}
\providecommand{\corollaryname}{Corollary}
\providecommand{\definitionname}{Definition}
\providecommand{\factname}{Fact}
\providecommand{\lemmaname}{Lemma}
\providecommand{\remarkname}{Remark}
\providecommand{\theoremname}{Theorem}

\begin{document}
\global\long\def\defeq{\stackrel{\mathrm{{\scriptscriptstyle def}}}{=}}%
\global\long\def\norm#1{\left\Vert #1\right\Vert }%
\global\long\def\R{\mathbb{R}}%
\global\long\def\otilde{\widetilde{O}}%
\global\long\def\Cov{\textrm{Cov}}%

\global\long\def\bdiag{\mathbf{Diag}}%
\global\long\def\nnz{\mathrm{nnz}}%
\global\long\def\mh{\mathbf{H}}%
\global\long\def\mi{\mathbf{I}}%
\global\long\def\mv{\mathbf{V}}%
\global\long\def\mw{\mathbf{W}}%
\global\long\def\ms{\mathbf{S}}%
\global\long\def\mproj{\mathbf{P}}%
\global\long\def\msigma{\mathbf{\Sigma}}%
\global\long\def\ma{\mathbf{A}}%
 
\global\long\def\Rn{\mathbb{R}^{n}}%
\global\long\def\tr{\mathrm{Tr}}%
\global\long\def\poly{\mbox{poly}}%
\global\long\def\diag{\mathrm{diag}}%
\global\long\def\cov{\mathrm{Cov}}%
\global\long\def\E{\mathbb{E}}%
\global\long\def\P{\mathbb{P}}%
\global\long\def\Var{\mathrm{Var}}%
\global\long\def\rank{\mathrm{rank}}%
\global\long\def\Ent{\mathrm{Ent}}%
\global\long\def\vol{\mathrm{vol}}%
\global\long\def\spe{\mathrm{op}}%
\global\long\def\op{\mathrm{op}}%

\title{Reducing Isotropy and Volume to KLS:\linebreak{}
Faster Rounding and Volume Algorithms}
\author{He Jia\thanks{Georgia Tech, \{hjia36, aladdha6, vempala\}@gatech.edu}\and
Aditi Laddha\footnotemark[1] \and Yin Tat Lee\thanks{University of Washington and Microsoft Research, yintat@uw.edu} \and
Santosh S. Vempala\footnotemark[1]}

\maketitle

\begin{abstract}
We show that the volume of a convex body in $\R^{n}$ in the general
membership oracle model can be computed to within relative error $\varepsilon$
using $\widetilde{O}(n^{3.5}\psi^{2} + n^3/\varepsilon^{2})$ oracle queries,
where $\psi$ is the KLS constant. With the current bound of $\psi=\widetilde{O}(1)$,
this gives an $\widetilde{O}(n^{3.5} + n^3/\varepsilon^{2})$ algorithm, improving
on the Lov\'{a}sz-Vempala $\widetilde{O}(n^{4}/\varepsilon^{2})$
algorithm from 2003. The main new ingredient is an $\widetilde{O}(n^{3}\psi^{2})$
algorithm for isotropic transformation of a well-rounded convex body; we apply this iteratively to isotropicize 
a general convex body.
Following this, we can apply
the $\widetilde{O}(n^{3}/\varepsilon^{2})$ volume algorithm of Cousins
and Vempala for well-rounded convex bodies. We also give an efficient
implementation of the new algorithm for convex polytopes defined by
$m$ inequalities in $\R^{n}$: polytope volume can be estimated in
time $\widetilde{O}(mn^{c}/\varepsilon^{2})$ where $c<3.7$ depends
on the current matrix multiplication exponent and improves on the
previous best bound.
\end{abstract}

\section{Introduction}

Computing the volume is a fundamental problem from antiquity, playing
a central role in the development of fields such as integral calculus,
thermodynamics and fluid dynamics. Over the past several decades,
numerical estimation of the volumes of high-dimensional bodies that
arise in applications has been of great interest. To mention one example
from systems biology, volume has been proposed as a promising parameter
to distinguish between the metabolic networks of normal and abnormal
individuals \cite{HCTFV2017} where the networks are modeled as very
high dimensional polytopes.

As an algorithmic problem for convex bodies in $\R^{n}$, volume computation
has a four-decade history. Early results by Bar\'{a}ny and by Bar\'{a}ny
and F\"{u}redi showed that any deterministic algorithm for the task
is doomed to have an exponential complexity, even to approximate the
volume to within an exponentially large factor. Then came the stunning
breakthrough of Dyer, Frieze, and Kannan showing that with randomization,
the problem can be solved in full generality (the membership oracle
model) to an arbitrary relative error $\varepsilon$ in time polynomial
in $n$ and $1/\varepsilon$. They used the Markov chain Monte Carlo
method and reduced the volume problem to sampling uniformly from a
sequence of convex bodies, and showed that the sampling itself can
be done in polynomial-time. Subsequent progress on the complexity
of volume computation has been accompanied by the discovery of several
techniques of independent interest, summarized below. The current
best complexity of $\widetilde{O}(n^{4}/\epsilon^{2})$\footnote{$\widetilde{O}$ suppresses polylogarithmic terms. $O^{*}$suppresses
dependence on error parameters as well as polylogarithmic terms.} for general
convex bodies is achieved by the 2003 algorithm of Lov\'{a}sz and
Vempala.

\begin{table}[h]
\centering{}%
\begin{tabular}{|c|l|c|}
\hline 
Year/Authors & New ingredients & Steps\tabularnewline
\hline 
\hline 
1989/Dyer-Frieze-Kannan \cite{DyerFK89} & Everything & $n^{23}$\tabularnewline
\hline 
1990/Lov\'{a}sz-Simonovits \cite{LS90} & Better isoperimetry & $n^{16}$\tabularnewline
\hline 
1990/Lov\'{a}sz \cite{L90} & Ball walk & $n^{10}$\tabularnewline
\hline 
1991/Applegate-Kannan \cite{ApplegateK91} & Logconcave sampling & $n^{10}$\tabularnewline
\hline 
1990/Dyer-Frieze \cite{DyerF90} & Better error analysis & $n^{8}$\tabularnewline
\hline 
1993/Lov\'{a}sz-Simonovits \cite{LS93} & Localization lemma & $n^{7}$\tabularnewline
\hline 
1997/Kannan-Lov\'{a}sz-Simonovits \cite{KLS97} & Speedy walk, isotropy & $n^{5}$\tabularnewline
\hline 
2003/Lov\'{a}sz-Vempala \cite{LV2} & Annealing, hit-and-run & $n^{4}$\tabularnewline
\hline 
2015/Cousins-Vempala \cite{CV2015} (well-rounded) & Gaussian Cooling & $n^{3}$\tabularnewline
\hline 
\end{tabular}\caption{\label{tab:volume} The complexity of volume estimation. The complexity
of all algorithms has an $\epsilon^{-2}$ factor. Each step uses $\widetilde{O}(n)$
bit of randomness and $\widetilde{O}(n^{2})$ arithmetic operations.
The last algorithm is for well-rounded convex bodies.}
\end{table}

The main subroutine for volume computation is sampling. Sampling is
done by random walks, notably the ball walk and hit-and-run. The rate
of convergence of random walks is determined by their conductance.
For a Markov chain with state space $\Omega$, transition function
$P_{u}(.)$ and stationary density $Q$, the conductance is
\[
\phi=\min_{A\subset\Omega:Q(A)\le\frac{1}{2}}\frac{\int_{u\in A}P_{u}(\Omega\setminus A)\,dQ(u)}{Q(A)}
\]
i.e., the minimum conditional escape probability with the stationary
density (probability of crossing from a set to its complement starting
from the stationary density in the set). The analysis of the ball
walk done in \cite{KLS97} shows that the conductance can be bounded
in terms of the isoperimetry of the stationary distribution, a purely
geometric parameter. For an $n$-dimensional measure $\nu,$ the reciprocal of its KLS
constant $\psi_{\nu}$ (i.e., the isoperimetric or Cheeger constant)
is
\[
\frac{1}{\psi_{\nu}}=\inf_{A\subset\R:Q(A)\le\frac{1}{2}}\frac{\nu_{n-1}(\partial A)}{\nu(A)}
\]
where $\partial A$ is the boundary of the set $A$, and $\nu_{n-1}$
is the induced $(n-1)$-dimensional measure. The conductance of the
Markov chain reduces to the isoperimetry of its stationary density\footnote{We use $\gtrsim$ to denote ``LHS greater than a constant factor
times RHS''.}:
\[
\phi\apprge\frac{1}{\psi\cdot n}
\]
with $\psi$ being the KLS constant of the stationary density of the
Markov chain. This implies a mixing rate of $O(n^{2}\psi^{2})$. Thus,
bounding the KLS constant becomes critical, and this consideration
originally motivated the KLS conjecture. With many unexpected connections
and applications since its formulation, the conjecture has become
a central part of asymptotic convex geometry and functional analysis.
(A distribution is said to be isotropic if it has zero mean and identity
covariance, see Def. \ref{def:isotropy}).
\begin{conjecture}[KLS Conjecture \cite{KLS95}]
The KLS constant of any isotropic logconcave density in any dimension
is bounded by an absolute constant. Equivalently, for a logconcave
density $q$ with covariance matrix $A$, we have $\psi_{q}\lesssim\norm A_{\op}^{1/2}$.
\end{conjecture}

Another equivalent formulation of the conjecture is that for any logconcave
density, a halfspace induced subset achieves the extremal isoperimetry
up to an absolute constant. The paper \cite{KLS95} also showed that
\[
\psi_{q}=O\left(\sqrt{\tr A}\right)
\]
which is $O(\sqrt{n})$ for isotropic logconcave densities. This implies
mixing of the ball walk (and hence sampling) in $n^{3}$ steps from
a warm start in an isotropic convex body $K$ containing a unit ball.
But how to find an isotropic transformation and maintain a warm start?
They do this with two essential ingredients: (i) a bound of $O(n^{2}R^{2})$
on the mixing time where $R^{2}=\E_{K}(\norm{x-\bar{x}}^{2})$ (ii)
interleaving the volume algorithm with isotropic transformation: the
algorithm starts with a simple isotropic body (like a ball) and then
goes through a sequence of $\widetilde{O}(n)$ convex bodies, maintaining
well-roundedness, i.e., $R=\widetilde{O}(\sqrt{n})$, and computing
ratios of volumes of consecutive bodies as well as isotropic transformation,
via sampling; a random sample from the current phase serves as a warm
start for the next phase. As $\widetilde{O}(n)$ samples are needed
to estimate each ratio and for the isotropic transformation, this
implies a volume/rounding algorithm of complexity $n^{3}\cdot n\cdot n=O^{*}(n^{5}).$

The algorithm of \cite{LV2} improves on this as follows: (a) they
separated the isotropic transformation from volume estimation by giving
an $\widetilde{O}(n^{4})$ algorithm for isotropic transformation
and (b) they replaced the sequence of convex bodies with a sequence
of logconcave densities (``simulated annealing''), specifically
exponential densities restricted to convex bodies; they showed that
a sequence of $\widetilde{O}(\sqrt{n})$ densities suffices, while
still maintaining a warm start. This reduced the overall complexity
to $n^{3}\cdot\sqrt{n}\cdot\sqrt{n}=O^{*}(n^{4})$. (Note that by
a simple variance analysis, the number of samples needed per phase
grows linearly with the number of phases \cite{DyerF90}.) The total
number of samples used in the algorithm is $\widetilde{O}(n)$, so
further improvements would require faster sampling. The concluding
remark from \cite{LV2} says:

\emph{``There is one possible further improvement on the horizon.
This depends on a difficult open problem in convex geometry, a variant
of the ``Slicing Conjecture'' }\cite{KLS95}\emph{. If this conjecture
is true... could perhaps lead to an $O^{*}(n^{3})$ volume algorithm.
But besides the mixing time, a number of further problems concerning
achieving isotropic position would have to be solved.''}

Since then there has been much progress on the KLS conjecture; following results by Lee and Vempala~\cite{LeeV17KLS,lee2024eldan}, Chen~\cite{chen2021almost}, 
Klartag and
Lehec \cite{Klartag2022,JLV2022}, the current best bound is 
$O(\sqrt{\log n})$ by Klartag \cite{klartag2023logarithmic}.
\begin{thm}[\cite{klartag2023logarithmic}]
For any isotropic logconcave density $p$, we have $\psi_{p}\lesssim \sqrt{\log n}.$
\end{thm}
Despite these improvements to the KLS constant and the mixing rate
of the ball walk, the complexity of volume computation remained at
$O^*(n^{4})$. Even outputting the first random point needs $n^{4}$ oracle
queries, since the mixing rate improvement is only from a warm start
in an isotropic body.

Progress in a different line, without using the KLS conjecture came
in 2015. Cousins and Vempala gave a volume algorithm with complexity
$O^{*}(n^{3})$ for any \emph{well-rounded} convex body, i.e., they
assume the input body contains the unit ball and is mostly contained
in a ball of radius $R=\widetilde{O}(\sqrt{n})$. This is a weaker
condition than (approximate) isotropic position, which requires that
the covariance matrix of $K$ is (close to) the identity. Their Gaussian
Cooling algorithm uses a sequence of Gaussians restricted to the body,
starting with a Gaussian of small variance almost entirely contained
in the body and flattening it to a near-uniform distribution. Notably,
they bypass the KLS conjecture, needing isoperimetry only for the
special case of the Gaussian density restricted to a convex body,
for which $\psi=O(1)$. The main open problem remaining after their
work was to find a faster algorithm to make the body well-rounded
(or isotropic). An improved rounding algorithm would directly imply
a faster volume algorithm.

\subsection{Main results}

We give a new algorithm for the isotropic transformation of any given
convex body. At the core is an algorithm for isotropic transformation of a well-rounded convex body, 
which we state first. 
\begin{thm}[Well-rounded to Isotropic]
There is an algorithm that takes as input a well-rounded convex body $K$, i.e., satisfying $B(0,1)\subseteq K$
with $\E_{K}\norm x^{2}=\widetilde O(n)$, and with high probability, finds an affine transformation
$T$ using $\widetilde{O}(n^{3}\psi^{2})$ oracle calls, s.t. $TK$
is $2$-isotropic.
\end{thm}
For a general convex body, using a standard method of considering a sequence of growing balls intersected 
with the body, we get faster algorithm for isotropic transformation of any convex body.
With the current bound of $\psi=\widetilde{O}(1)$, the complexity
is $O^{*}(n^{3.5})$. We remark that the time per query is based on maintaining an affine
transformation and computing a matrix vector product. It is standard
for all volume algorithms in the general oracle model.

\begin{thm}[Isotropy]
\label{thm:rounding}There is a randomized algorithm that takes as
input a convex body $K\subset\R^{n}$ given by a membership oracle
with initial point $x_{0}\in K$, bounds $r,R>0$ s.t.\[
x_{0}+rB_{n}\subseteq K\subseteq RB_{n}
\]
 and with probability $1-\delta$, computes an affine transformation
$T$ s.t., $TK$ is in near-isotropic position, i.e., for $x$ sampled
uniformly from $TK$, 
\[
I\preccurlyeq\cov(xx^{\top})\preccurlyeq2I.
\]
The algorithm uses $O(n^{3.5}\psi^{2}\log^{O(1)}(Rn/r)\log(1/\delta))$
membership oracle queries, which is $\widetilde{O}(n^{3.5})$ for the
current KLS constant bound of $\psi=\widetilde{O}(1)$. The time complexity
is $O(n^{2})$ per oracle query.
\end{thm}

Faster rounding implies faster volume computation for general convex bodies. 
\begin{cor}[Volume]
 The complexity of computing the volume of a convex body $K\subset\R^{n}$
given by a well-defined membership oracle with probability at least
$1-\delta$ to within relative error $\varepsilon$ is 
\[
O(n^{3.5}\psi^{2}\log^{O(1)}(Rn/r)\log(1/\delta))+O(n^3/\epsilon^2)=\widetilde{O}(n^{3.5}+n^{3}/\varepsilon^{2}).
\]
The time complexity is $O(n^{2})$ per oracle query.
\end{cor}

For an explicit convex polytope $Ax\le b$, with $A\in R^{m\times n}$,
since each oracle query takes $O(mn)$ time, this immediately gives
a bound of $\widetilde O(mn^{4.5} + mn^4/\epsilon^2)$, matching the current best bound in \cite{mangoubi2019faster}. 
There is a different approach
\cite{lee2018convergence} giving a time bound of $O^{*}(m^{2}n^{\omega-\frac{1}{3}})$ (here $\omega$ is the matrix multiplication exponent),
which is better when the number of facets is close to linear in the
dimension. Here we give an efficient implementation of our new algorithm
that improves significantly on the runtime. Let $\omega(a,b,c)$ be
the exponent of complexity of multiplying an $n^{a}\times n^{b}$
matrix with an $n^{b}\times n^{c}$ matrix, e.g., $\omega(1,1,1)=\omega\le2.371552$
\cite{duan2023faster,AlmanW21,williams2024new}. The following result further improves the complexity of
polytope volume computation for large $m$.

\begin{restatable}[Polytope Volume]{thm}{polytopevol}\label{thm:polytope}

The volume of a convex polytope $\left\{ x:Ax\leq b\right\} \subset\R^{n}$
defined by $m$ inequalities can be computed with high probability
to within relative error $\varepsilon$ using fast matrix multiplication
in time $\widetilde{O}(mn^{c}/\varepsilon^{2})$ where $c=\omega(1,1,3.5)-1<3.7$.

\end{restatable}

\subsection{Approach}

The algorithm is based on the following ideas. First, if the covariance
matrix is skewed (i.e., some eigenvalues are much larger than others),
this can be detected via a small random sample. The roundness can
then be improved by scaling up the subspace of small eigenvalues.
But how to sample a highly skewed convex body? To break this chicken-and-egg
problem, we use a sequence of convex bodies obtained by intersecting
the original set with balls of increasingly larger radii: 
\[
K_{t}=K\cap B(0,t),
\]
with $t$ starting at $r$ and increasing by a factor of $(1+\frac{1}{\sqrt{n}})$ till it reaches $R.$ In each
outer iteration, we compute a transformation that makes the set $K_{t}$
(nearly) isotropic. When $K_{t}$ is isotropic, we show that $K_{(1+\frac{1}{\sqrt{n}})t}$
is well-rounded (Lemma \ref{lem:round-3.5})\footnote{A conference version of this paper \cite{jia2021reducing} incorrectly claimed a factor of $2$ scaling works, in place of the $(1+(1/\sqrt{n}))$ factor.}. Hence, the trace of the
covariance of $K_{(1+\frac{1}{\sqrt{n}})t}$ is $\widetilde{O}(n)$. However, its eigenvalues could
be widely varying. To make $K_{(1+\frac{1}{\sqrt{n}})t}$ isotropic, we will estimate the
larger eigenvalues using a small random sample (Lemma \ref{lem:chernoff_app}).
The second idea is that scaling up the small eigenvalues nearly doubles
the size of the ball contained inside (Lemma \ref{lem:innerball}),
while having a mild effect on the higher norms of the covariance (Lemma
\ref{lem:basic_property_isotropization}). The latter concept is where
KLS comes in. We show that if the KLS constant for isotropic logconcave
distributions in $\R^{n}$ is bounded by $\psi_{n}^{2}\lesssim n^{\frac{1}{p}}$
for all $n$, then for any logconcave density $q$ with covariance
matrix $A$ (not necessarily identity), the KLS constant is bounded
as $\psi_{q}^{2}\lesssim\norm A_{p}\log n$ (Theorem \ref{thm:KLS_ani}).
As we outline below, this improves the sampling time in each inner
iteration.

In the beginning, when the ball contained inside $K_{t}$ is small,
we will use a few samples to get a coarse estimate of the larger eigenvalue
directions and scale up the orthogonal subspace. The sampling time
is higher since the roundness parameter of $K_{t}$ is higher. The
time per sample is roughly $n^{2}\frac{\norm A_{p}}{r^{2}}$ when
$r$ is the radius of the ball inside $K_{t}$ and $A$ is the covariance
matrix for the uniform density $q$ on $K_{t}$. As we increase $r$,
by a constant factor in each step, the norm of the covariance grows
much more slowly, and so the sampling time decreases. Meanwhile, we
use more samples in each step, roughly $r^{2}$, and this trade-off
keeps the overall time at $n^{3+\frac{1}{p}}\simeq n^{3}\psi^{2}.$
During the process we need a warm start for each phase; we achieve
this in $\widetilde{O}(n^{3})$ steps using the Gaussian Cooling algorithm.
It is known that $\psi_{q}^{2}\lesssim\norm A_{p}$ holds for $1/p=o(1)$,
giving the $\widetilde{O}(n^{3})$ complexity for rounding in inner iteration. 

\section{Preliminaries}

For a positive definite matrix $n\times n$ matrix $A$, we will use the following matrix norms:
\begin{enumerate}
    \item the operator or spectral norm, $\norm A_{\op}$
    \item the Frobenius norm, $\norm A_{F}\defeq\tr(A^{2})^{1/2}$
\item the $p'th$ norm for $p\ge1$, defined
as
\[
\norm A_{p}\defeq\left(\sum_{i=1}^{n}|\lambda_{i}|^{p}\right)^{1/p}
\]
with $p=1$ being $\tr A$ and $p=2$ being $\norm A_{F}.$
\end{enumerate}
\begin{defn}
A function $f:\R^{n}\rightarrow\R_{+}$ is \emph{logconcave} if its
logarithm is concave along every line, i.e., for any $x,y\in\R^{n}$
and any $\lambda\in[0,1]$,
\begin{equation}
f(\lambda x+(1-\lambda)y)\ge f(x)^{\lambda}f(y)^{1-\lambda}.\label{eq:logcon}
\end{equation}
\end{defn}

Many common probability distributions have logconcave densities, e.g.,
Gaussian, exponential, logistic, and gamma distributions; indicator
functions of convex sets are also logconcave. Logconcavity is preserved
by product, min and convolution (in particular marginals of logconcave
densities are logconcave).
\begin{defn}
\label{def:isotropy}A distribution $D$ is said to be \emph{isotropic}
if
\[
\E_{D}(x)=0\quad\mbox{ and \ensuremath{\quad\E_{D}(xx^{\top})=I}}.
\]
We say a convex body is isotropic if the uniform distribution over
it is isotropic. Any distribution with bounded second moments can
be brought to an isotropic position by an affine transformation. We
say that a convex body $K$ is \emph{$(r,R)$-rounded }if it contains
a ball of radius $r$, and its covariance matrix $A$ satisfies
\[
\tr A=\E_{K}(\norm{x-\E_K(x)}^{2})=R^{2}.
\]
For an isotropic body we have $A=I$ and hence $R^{2}=n.$ We will
say $K$ is well-rounded if $r=1$ and $R=\widetilde{O}(\sqrt{n})$.
Clearly, any isotropic convex body is also well-rounded, but not vice
versa. A distribution is said to be $\alpha$-isotropic if $I\preceq\E_{D}(xx^{\top})\preceq\alpha I$.

Random points from isotropic logconcave densities have strong concentration
properties. We mention three that we will use.
\end{defn}

\begin{lem}[{\cite[Theorem 5.17]{LV07}}]
\label{lem:volume_outside_ball}For any $t\ge1$, and any logconcave
density $p$ in $\R^{n}$ with covariance matrix $A$,
\[
\P_{x\sim p}\left(\|x-\E x\|_{2}\geq t\cdot\sqrt{\tr A}\right)\leq\exp(-t+1).
\]
\end{lem}

\begin{lem}[{\cite[Theorem 1.1]{Paouris2006}}]
\label{lem:volume_outside_ball2}For any $t>0$, and any isotropic
logconcave density $p$ in $\R^{n}$, there is a constant $c$ such
that
\[
\P_{x\sim p}(\|x\|_{2}-\sqrt{n}\geq t\sqrt{n})\leq\exp(-ct\sqrt{n}).
\]
\end{lem}

\begin{lem}[\cite{Adamczak2010,srivastava2013covariance}]
\label{lem:covariance} For any isotropic logconcave distribution
$\nu$ in $\R^{n}$, with probability at least $1-\delta$, the empirical
mean and covariance
\[
\overline{x}=(1/N)\sum_{i=1}^{N}x_{i},\quad\mbox{ and }\quad X=\frac{1}{N}\sum_{i=1}^{N}(x_{i}-\overline{x}_{i})(x_{i}-\overline{x}_{i})^{T}
\]
of $N=O(n\log(1/\delta)/\varepsilon^{2})$ i.i.d. random samples $x_{i}\sim\nu$
satisfy
\[
\norm{\overline{x}}_{2}\le\varepsilon,\quad\norm{X-I}_{\op}\le\varepsilon.
\]
\end{lem}

The convergence of Markov chains is established by showing the the
$t$-th step distribution $Q_{t}$ approaches the steady state distribution
$Q$. We will use the total variation distance $d_{TV}$ for this.
We also need a notion of a warm start.
\begin{defn}[Warm start]
\label{def:warm_start} We say that a starting distribution $Q_{0}$
is $M$-warm for a Markov chain with unique stationary distribution
$Q$ if its $\chi$-squared distance is bounded by $M$:
\[
\chi^{2}(Q_{0},Q)=\E_{Q}\left(\frac{dQ_{0}(u)}{dQ(u)}-1\right)^{2}\le M.
\]
Note that $\chi^{2}(Q_{0},Q)+1=\E_{Q}\left(\frac{dQ_{0}(u)}{dQ(u)}\right)^{2}=\E_{Q_{0}}\frac{dQ_{0}(u)}{dQ(u)}$.
\end{defn}

Our algorithms will use the ball walk for sampling. In a convex body
$K$, the \emph{ball walk} with step-size $\delta$ is defined as
follows: at the current point, $x\in K$, pick $y$ uniformly in $B(x,\delta)$,
the ball of radius $\delta$ around $x$; if $y\in K$, go to $y$
(else, do nothing).

The next theorem is a fast sampler for distributions with KLS constant
$\psi_{q}$ given a warm start.
\begin{thm}[\cite{KLS97}]
\label{thm:warm_start}Let $K$ be a convex body containing the unit
ball. Using the ball walk with step size $\frac{1}{\sqrt{n}}$ in
$K$ from an $M$-warm start $Q_{0}$, the number of steps to generate
a nearly independent point within $TV$ distance or $\chi^{2}$-distance
$\varepsilon$ of the uniform stationary density $Q$ in $\R^{n}$
is $O(n^{2}\psi_{q}^{2}\log^{O(1)}(nM/\epsilon)).$
\end{thm}

The next lemma connects the KLS constant for isotropic distributions
to that of general distributions. We give a proof in Section \eqref{sec:KLS}.
\begin{lem}
\label{lem:KLS-anisotropic}Let $\psi_{n}$ be a bound on the KLS
constant for isotropic logconcave densities and $\psi_{q}$ be the
bound for a logconcave density $q$ (not necessarily isotropic) with
covariance matrix $A_{q}$, both in $\R^{n.}$. If $\psi_{n}^{2}=\Theta\left(n^{1/p}\right)$
for all $n$, then for any logconcave density $q$ with covariance
$A_q$, we have $\psi_{q}^{2}=O(\norm {A_q}_{p}\log^{O(1)}n)$.
\end{lem}

We will also use a fast sampler for well-rounded convex bodies that
does not require a warm start.
\begin{thm}[Sampling a well-rounded convex body \cite{CV2015,CousinsV18}]
\label{thm:sampling} There is an algorithm that, for any $\varepsilon>0,p>0$,
and any convex body $K$ in $\R^{n}$ that contains the unit ball
and has $\E_{K}(\|X\|^{2})=R^{2}$, with probability $1-\delta$,
generates random points from a density $\nu$ that is within total
variation distance $\varepsilon$ from the uniform distribution on
$K$. In the membership oracle model, the complexity of each random
point, including the first, is, 
\[
O\left((R^{2}n^{2}+n^{3})\log n\log^{2}\frac{n}{\varepsilon}\log\frac{1}{\delta}\right)=\widetilde{O}(n^{3}+R^{2}n^{2}).
\]
\end{thm}

Finally, we use the algorithm from \cite{CousinsV18,CV2015} for computing
the volume of a well-rounded convex body.
\begin{thm}[Volume of a well-rounded convex body \cite{CV2015,CousinsV18}]
\label{thm:well_rounded_volume} There is an algorithm that, for
any $\varepsilon>0,\delta>0$ and convex body $K$ in $\R^{n}$ that
contains the unit ball and has $\E_{K}(\|X\|^{2})=O(n)$, with probability
$1-\delta$, approximates the volume of $K$ to within relative error
$\varepsilon$ and has complexity 
\[
O\left(\frac{n^{3}}{\varepsilon^{2}}\cdot\log^{2}n\log^{2}\frac{1}{\varepsilon}\log^{2}\frac{n}{\varepsilon}\log\frac{1}{\delta}\right)=\widetilde{O}(n^{3}\epsilon^{-2})
\]
in the membership oracle model.
\end{thm}

\paragraph{Computational Model.}

We use the most general membership oracle model for convex bodies,
which comes with bounds $r,R$ guaranteeing that the input convex
body $K$ has inner radius $r$ and outer radius $R$ and allows membership
queries to $K$. As established in the literature on volume computation,
the number of arithmetic operations is bounded by $O(n^{2})$ per
oracle query, and all arithmetic operations can be done using only
a polylogarithmic number of additional bits. We mention a few familiar
technical difficulties whose solution is well-documented in the literature.
First, the sampling algorithms produce points from \emph{approximately
}the correct distribution. Second, samples produced in a sequence
are not completely independent. Third, all our algorithms are randomized
(they have to be in the oracle model). For the first two issues, we
refer to \cite{lovasz2006simulated}; briefly, the approximate distribution
is handled by a trick called ``divine intervention'', where one
can view the sampled distribution as being the correct one with large
probability and an incorrect one with a small failure probability.
Since the complexity dependence on proximity to the target is logarithmic,
this leads to a controllable overall failure probability. The near-independence
is handled as follows: first, we maintain parallel independent threads
of samples, which start as completely random points (e.g., from a
Gaussian) and remain independent throughout. For estimates computed
in different sequential phases (e.g., ratios of integrals, or affine
transformations), the degree of dependence is explicitly bounded using
the tools developed in \cite{lovasz2006simulated}. For the failure
probability, we note that the overhead for failure probability $\delta$
is $O(\log(1/\delta))$. In the rest of the paper, for convenience,
we say WHP to mean with probability $1-\delta$ incurring only $O(\log(1/\delta))$
overhead in the complexity.

\section{Algorithm and Analysis}

The algorithm considers a sequence of balls of radii that grow by a factor of $(1+\frac{1}{\sqrt{n}})$, and
in each iteration makes the intersection of the current ball with
the convex body nearly isotropic.

\begin{algorithm}[H]
\caption{Iterative Isotropization}
\label{alg:iterative_rounding_general}

\begin{algorithmic}[1]

\Procedure{\textsc{$\mathsf{IterativeIsotropization}$}}{$K\subset\Rn$,
$r>0$, $R>0$}

\State \textbf{Assumption:} $B(0,r)\subseteq K\subseteq B(0,R)$

\State Define $K_{t}=K\cap B(0,t).$

\State $t\leftarrow r,x\leftarrow0,T\leftarrow\frac{1}{r}I$\Comment{Current
slice $t$ and normalization}

\While{$t<R$}

\State \textsc{$(x,T)\leftarrow\mathsf{Isotropize}(K_{t}-x,T)$}

\State $t\rightarrow (1+\frac{1}{\sqrt{n}})t$

\EndWhile

\EndProcedure

\end{algorithmic}
\end{algorithm}

The main subroutine, Isotropize, is a procedure to compute an isotropic
transformation of a well-rounded body. Given a convex body $K$ and
a transformation $T$ such that $TK$ is well-rounded, Isotropize
outputs a point $x$ and a transformation $T'$ such that $T'\left(K-x\right)$
is $2$-isotropic. We address this in the next section, then come
back to the general analysis.

\subsection{Inner loop: Isotropic transformation of a well-rounded body}

We begin with the algorithm. In each iteration, the inner ball radius
grows by a constant factor, while the $p$-norm of the covariance
grows much more slowly. As a result, we can afford to sample more
points with each iteration and thereby get progressively better approximations
to the isotropy: in the first step, the number of samples is only
$\widetilde{O}(1)$ while by the end the number of samples is $\widetilde{O}(n)$.

\begin{algorithm}[!th]
\caption{Isotropize}
\label{alg:iterative_rounding}

\begin{algorithmic}[1]

\Procedure{$\mathsf{Isotropize}$}{$K\subset\Rn$, $T_{0}\in\R^{n\times n}$}

\State \textbf{Assumption:} $B(0,\frac{1}{4})\subset T_{0}K$ and
$\E_{x\sim T_{0}K}\|x\|^{2}=\widetilde O(n)$.

\State Use Gaussian cooling to sample a point $x\in T_{0}K$.

\State $r\leftarrow\frac{1}{4},T\leftarrow T_{0}$.\Comment{We maintain
$\{x:\|x\|_{2}\leq r\}\subseteq TK.$}

\While{$2^{10}r^{2}\log^2 n\leq n$}

\State $k\leftarrow c\cdot r^{2}\log^{7}n$ for a constant $c$.

\State Sample points $x_{1},\cdots,x_{2k}$ from $TK$ using the
ball walk with initial point $x$.

\State $\widehat{A}\leftarrow\frac{1}{k}\sum_{i=1}^{k}(x_{i}-\widehat{x})(x_{i}-\widehat{x})^{\top}$
where $\widehat{x}=\frac{1}{k}\sum_{i=1}^{k}x_{i+k}$. \Comment{Estimate
covariance.}

\State $M\leftarrow I_{n}+P_{\widehat{A}}$ where $P_{\widehat{A}}$
is the projection to the subspace spanned by eigenvectors of $\widehat{A}$
with eigenvalue at most $\lambda=n$.

\State $T\leftarrow MT$, $r\leftarrow2(1-\frac{1}{\log n})\cdot r$.

\EndWhile

\State Sample $O(n\log^2 n)$ points to compute the mean $\widehat{x}$ and
empirical covariance matrix $\widehat{A}$ such that $\widehat{A}^{-\frac{1}{2}}(K-\widehat{x})$
is $2$-isotropic.

\State \textbf{Return} $(\widehat{x},\widehat{A}^{-\frac{1}{2}})$.

\EndProcedure

\end{algorithmic}
\end{algorithm}

\begin{figure}
\centering{}\includegraphics[scale=0.4]{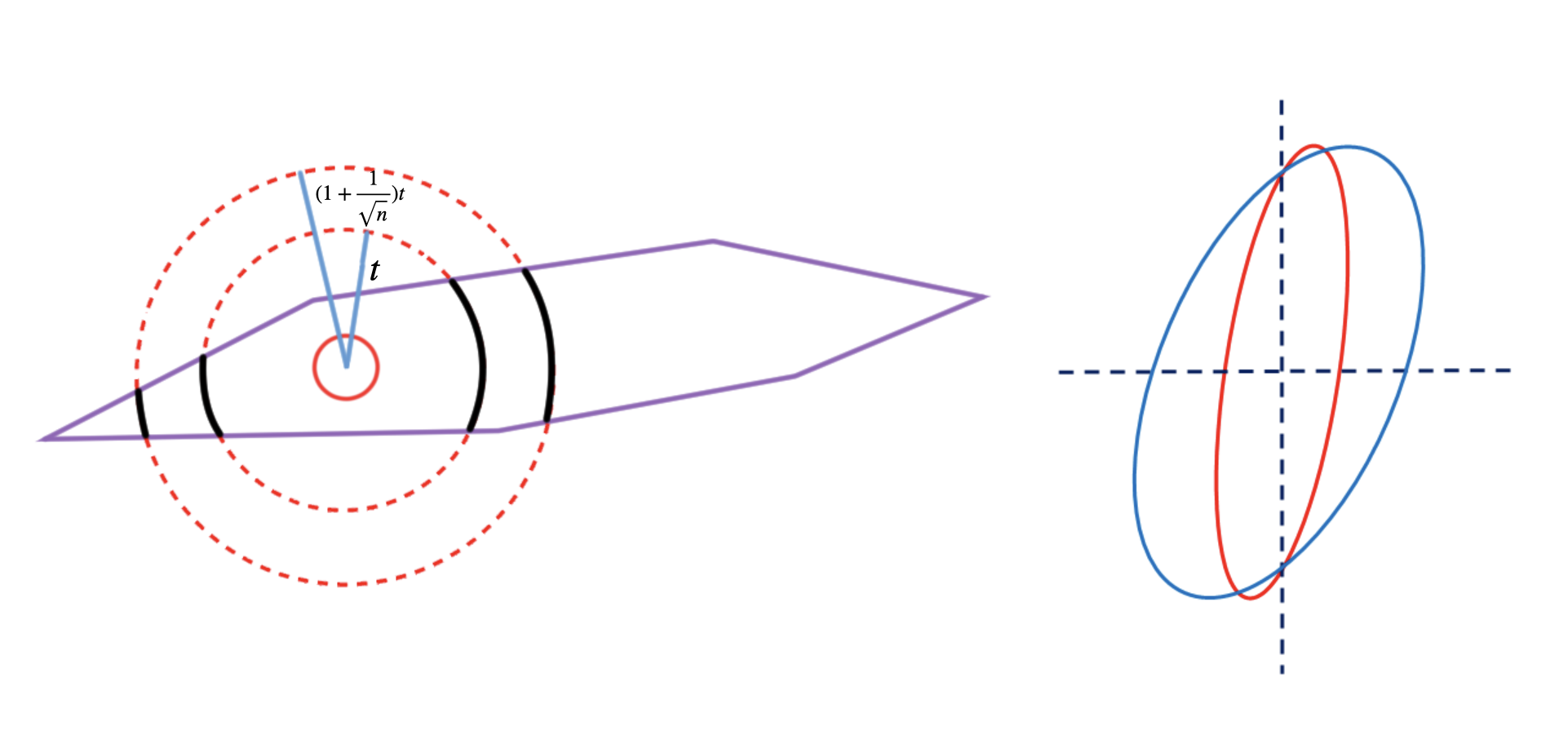}
\caption{(a) Algorithm $\mathsf{IterativeIsotropization}$ uses balls of growing
radii (b) Algorithm $\mathsf{Isotropize}$ scales up the estimated
``small eigenvalue'' subspace in each iteration.}
\end{figure}

First, we show that the $p$-th norm of the covariance matrix $\|A\|_{p}$
remains bounded. Although we only use an eigenvalue threshold of $\lambda=n$
in the algorithm, we will prove a slightly more general statement
below, assuming the eigenvalue threshold is $\lambda=n/r^{\alpha}$.
\begin{lem}
\label{lem:basic_property_isotropization} Let $A_{j}$ be the covariance
matrix of $TK$ and $r_{j}$ be the inner radius at the $j$-th iteration
of algorithm $\mathsf{Isotropize}$. For $\alpha\in[0,2/(p-1)]$,
we have
\begin{itemize}
\item the number of iterations is $O(\log n)$.
\item $\tr A_{j}=O(r_{j}^{2}\tr A_{0})$ for all j.
\item $\|A_{j}\|_{p}\le24nr_{j}^{(\frac{2+\alpha}{p}-\alpha)}\log^{1+\frac{2}{p}} n$ for
all j.
\end{itemize}
\end{lem}

\begin{proof}
In each iteration, the algorithm increases $r_{j}$ by a factor of
$2(1-\frac{1}{\log n})$. Note that the algorithm starts with $r_{0}=\frac{1}{4}$
and ends before $r_{j}\log n\leq\sqrt{n}$. Hence, it takes less than
$O(\log n)$ iterations.

To bound $\tr A_{j}$, we let $P_{j}$ be the projection matrix at
$j$-th iteration Then the transformation at $j$-th iteration is
$M_{j}=I+P_{j}$. We have
\begin{equation}
A_{j}=M_{j}A_{j-1}M_{j}\label{eq:Ak_update}
\end{equation}
Since $P_{j}$ is a projection matrix, we have
\begin{align*}
\tr A_{j} & =\tr\left(A_{j-1}^{\frac{1}{2}}M_{j}^{2}A_{j-1}^{\frac{1}{2}}\right)=\tr\left(A_{j-1}^{1/2}(I+3P_{j})A_{j-1}^{1/2}\right)\le4\tr A_{j-1}.
\end{align*}
Hence, $\tr A_{j}$ increases by a factor of at most $4$. Moreover,
since $r_{j}$ increases by a factor of $2(1-\frac{1}{\log n})$ we
have that $\frac{\tr A_{j}}{r_{j}^{2}}$ increases by a factor of
$(1-\frac{1}{\log n})^{-2}$ per iteration. Since there are $O(\log n)$
iterations, it increases by at most $O(1)$ in total. Hence, we have
$\tr A_{j}=O(r_{j}^{2}\tr A_{0})=O(r_{j}^{2}n\log^2 n)$ for all $j$.

To bound $\|A_{j}\|_{p}$, we note that initially $\|A_{0}\|_{p}\leq\tr A_{0}\leq4n$.
By \eqref{eq:Ak_update}, we have
\begin{align*}
\|A_{j}\|_{p} & =\|M_{j}A_{j-1}M_{j}\|_{p}\\
 & =\|A_{j-1}^{1/2}M_{j}^{2}A_{j-1}^{1/2}\|_{p}\\
 & =\|A_{j-1}^{1/2}(I+3P_{j})A_{j-1}^{1/2}\|_{p}.
\end{align*}
Hence, we have
\[
\|A_{j}\|_{p}\leq\|A_{j-1}\|_{p}+3\|A_{j-1}^{1/2}P_{j}A_{j-1}^{1/2}\|_{p}.
\]
Let $\lambda\defeq\frac{n}{r_{j}^{\alpha}}$. Since $P_{j}$ is the
projection of the eigenspace of $\widehat{A}$ with eigenvalues less
than $\lambda$, we have that
\[
P_{j}\preceq2\lambda(\widehat{A}+\lambda\cdot I)^{-1}.
\]
(To see this, we note that both the sides have the same eigenspace
and hence it follows from $1_{x\leq\lambda}\leq\frac{2\lambda}{x+\lambda}$.)
By Lemma \ref{lem:chernoff_app} with $\epsilon=\frac{1}{2}$, we
have that $\frac{1}{2}A_{j}-O(\frac{\log^{3}n\cdot\tr A_{j}}{k})\cdot I\preceq\widehat{A}.$
Using $\tr A_{j}=O(r_{j}^{2}n\log^2 n)$ and $k=\Omega(r_{j}^{2+\alpha}\log^{5}n)$
samples, we have 
\[
A_{j-1}\preceq2\widehat{A}+\frac{n}{r_{j}^{\alpha}}\cdot I\preceq2\widehat{A}+\lambda\cdot I.
\]
Hence, we have $P_{j}\preceq4\lambda(A_{j-1}+\lambda\cdot I)^{-1}$
and 
\[
\|A_{j}\|_{p}\leq\|A_{j-1}\|_{p}+12\lambda\|A_{j-1}^{1/2}(A_{j-1}+\lambda\cdot I)^{-1}A_{j-1}^{1/2}\|_{p}.
\]
Let $\lambda_{i}$ be the eigenvalues of $A_{j-1}$. Then, 
\[
\|A_{j-1}^{1/2}(A_{j-1}+\lambda\cdot I)^{-1}A_{j-1}^{1/2}\|_{p}^{p}=\sum_{i=1}^{n}\left(\frac{\lambda_{i}}{\lambda_{i}+\lambda}\right)^{p}\leq\frac{1}{\lambda^{p}}\sum_{\lambda_{i}\leq\lambda}\lambda_{i}^{p}+\sum_{\lambda_{i}\geq\lambda}1.
\]
The first term $\sum_{\lambda_{i}\leq\lambda}\lambda_{i}^{p}$ is
maximized when the small eigenvalues are exactly $\lambda_{i}=\lambda$.
In this case, there are 
\[
\frac{\tr A_{j-1}}{\lambda}\leq\frac{O(r_{j-1}^{2}n\log^2 n)}{\lambda}
\]
 small eigenvalues, i.e., of value at most $\lambda$. Hence, we have
\[
\frac{1}{\lambda^{p}}\sum_{\lambda_{i}\leq\lambda}\lambda_{i}^{p}\leq\frac{O(r_{j-1}^{2}n\log^2 n)}{\lambda}.
\]
For the second term, 
\[
\sum_{\lambda_{i}\geq\lambda}1\leq\frac{\tr A_{j-1}}{\lambda}\leq\frac{O(r_{j-1}^{2}n\log^2 n)}{\lambda}.
\]
Hence, $\|A_{j}\|_{p}$ is increased by
\[
O(\lambda)\cdot\left(\frac{r_{j-1}^{2}n\log^2 n}{\lambda}\right)^{1/p}=O\left(\frac{n}{r_{j}^{\alpha}}\right)\cdot\left(r_{j-1}^{2}r_{j}^{\alpha}\log^2 n\right)^{1/p}\le O\left(nr_j^{\frac{2+\alpha}{p}-\alpha}\log^{2/p} n\right).
\]
Since there are $O(\log n)$ iterations, $\|A_{j}\|_{p}$ is at most
$O(nr_j^{(\frac{2+\alpha}{p}-\alpha)}\log^{1+\frac{2}{p}} n)$.
\end{proof}
Next, we show that the inner radius increases by almost a factor of
$2$ every iteration, again for a general choice of eigenvalue threshold
$\lambda=n/r^{\alpha}.$
\begin{lem}
\label{lem:innerball} Algorithm $\mathsf{Isotropize}$ maintains
the invariant $\{\|x\|_{2}\leq r\}\subseteq TK$.
\end{lem}

\begin{proof}
First, we describe the idea of the proof. When the algorithm modifies
the covariance matrix $A$, it estimates the subspace of directions
with variance less than $n/r^{\alpha}$ and doubles them. Lemma \ref{lem:chernoff_app} with $\epsilon=1/2$
shows that 
\[
\widehat{A}=(1\pm\frac{1}{2})A\pm\frac{n}{4r^{\alpha}}\cdot I.
\]
This means, roughly speaking, that any direction that was not doubled
satisfies 
\[
\lambda_{i}\ge\sigma^{2}=\frac{n}{2r^{\alpha}}.
\]
So the current body contains an ellipsoid whose axis lengths along
the non-doubled directions (call this subspace $V$) are at least
$\sigma$ and at least $r$ in all directions (the body contains a
ball of radius $r$). Now the body is stretched by a factor of $2$
in the subspace $V^{\perp}$. We will argue that the resulting body
contains a ball of radius nearly $2r$. To see this, we argue that
any point in a ball of this radius is in the convex hull of the ball
of radius $\sigma$ in the subspace $V$ and a ball of radius $2r$
in the subspace $V^{\perp}.$

Now we proceed to the formal proof by induction. Initially, we have
$\{\|x\|_{2}\leq r\}\subset TK$ by the assumption on $K$. It suffices
to show this is maintained after each iteration. Let $T$ be the old
transformation and $T'$ be the new transformation during one iteration.
Let $r$ be the old radius during that iteration. By the invariant,
we know that $\{\|x\|_{2}\leq r\}\subset TK$ and hence
\[
\{\|M^{-1}x\|_{2}\leq r\}\subset T'K
\]
with $M=I+P_{\widehat{A}}$ (this only scales up the body by a factor
of $2$ in some directions). Let $A$ be the covariance matrix of
$TK$. Then $(I+P_{\widehat{A}})A(I+P_{\widehat{A}})$ is the covariance
matrix of $T'K$. Our goal is to show that $T'K$ contains a ball
of radius $2(1-\frac{1}{\log n})r$. From above, we have
\[
\{x^{\top}A^{-1}x\leq1\}\subset\{x^{\top}((I+P_{\widehat{A}})A(I+P_{\widehat{A}}))^{-1}x\leq1\}\subset T'K.
\]

Hence, we know that $\text{Conv}(\Omega_{1}\cup\Omega_{2})\subset T'K$
where 
\[
\Omega_{1}=\{x^{\top}(\frac{1}{r^{2}}I-\frac{3}{4r^{2}}P_{\widehat{A}})x\leq1\}\mbox{ and }\Omega_{2}=\{x^{\top}A^{-1}x\leq1\}.
\]
where for $\Omega_{1}$ we used that $(I+P)^{-2}=I-\frac{3}{4}P$
for any projection matrix $P$.

To prove that $T'K$ contains a larger ball, we take $x$ and write
it as a convex combination of $x_{1}$ and $x_{2}$ where $x_{1}\in\Omega_{1}$
and $x_{2}\in\Omega_{2}$. To do this, we define $P_{A}$ be the projection
to the subspace spans by eigenvectors of $A$ with eigenvalues at
most $\lambda/2^{8}\log^{2}n$ with 
\[
\lambda\defeq\frac{n}{r^{\alpha}}.
\]

For any $x$ with $\|x\|_{2}\leq2(1-\frac{1}{\log n})r$, we will
show that $x\in\text{Conv}(\Omega_{1}\cup\Omega_{2})$. To do this,
we let $t^2\defeq\frac{\|P_{A}x\|_{2}^{2}}{\|x\|_{2}^{2}}$ and write
$x$ as the convex combination,
\begin{equation}
x=t\left(\frac{1}{t}P_{A}x\right)+(1-t)\left(\frac{1}{1-t}(I-P_{A})x\right).\label{eq:x_decomp}
\end{equation}
For the second term in \eqref{eq:x_decomp}, we have
\[
\left(\frac{1}{1-t}(I-P_{A})x\right)^{\top}A^{-1}\left(\frac{1}{1-t}(I-P_{A})x\right)\leq\frac{\|(I-P_{A})x\|_{2}^{2}}{(1-t)^{2}\lambda/(2^{8}\log^{2}n)}=\frac{2^{8}\|x\|_{2}^{2}\log^{2}n}{\lambda}.
\]
Since $\|x\|_{2}^{2}\leq4r^{2}$ and $\lambda=\frac{n}{r^{\alpha}}$,
we have
\[
\frac{2^{8}\|x\|_{2}^{2}\log^{2}n}{\lambda}\leq\frac{2^{10}r^{2}\log^{2}n}{\frac{n}{r^{\alpha}}}=\frac{2^{10}r^{2+\alpha}\log^{2}n}{n}\leq1
\]

where the last inequality follows the assumption in the algorithm
that $\alpha=0$. Hence, 
\begin{equation}
\frac{1}{1-t}(I-P_{A})x\in\Omega_{2}.\label{eq:term_2_proof}
\end{equation}

For the first term in \eqref{eq:x_decomp}, we note that for any $\beta>0$,
we have
\begin{equation}
P_{\widehat{A}}\succeq\beta\lambda(\beta\lambda\cdot I+\widehat{A})^{-1}-\frac{\beta}{1+\beta}I.\label{eq:P_lower}
\end{equation}
(To see this, we note that both sides have the same eigenspace and
hence it follows from $1_{x\leq\lambda}\geq\frac{\beta\lambda}{x+\beta\lambda}-\frac{\beta}{1+\beta}$.)

By Lemma \ref{lem:chernoff_app} with $\epsilon=1$, we have 
\[
\widehat{A}\preceq2A+O\left(\frac{\log^{3}n\cdot\tr A}{k}\right)\cdot I.
\]
Using $\tr A=O(r^{2}n\log^2 n)$ (Lemma \ref{lem:basic_property_isotropization})
and $k=\Omega(r^{2+\alpha}\log^{7}n)$ (the algorithm description),
we have
\[
\widehat{A}\preceq2A+\frac{n}{2^{7}r^{\alpha}\log^{2}n}\cdot I\preceq2A+\frac{\lambda}{2^{7}\log^{2}n}\cdot I.
\]
Putting this into \eqref{eq:P_lower}, we have
\[
P_{\widehat{A}}\succeq\beta\lambda(\beta\lambda\cdot I+\frac{\lambda}{2^{7}\log^{2}n}+2A)^{-1}-\frac{\beta}{1+\beta}I.
\]
On the range of $P_{A}$, we have $A\preceq\frac{\lambda}{2^{8}\log^{2}n}\cdot I$
and hence
\[
(P_{A}x)^{\top}P_{\widehat{A}}(P_{A}x)\geq\left(\frac{\beta\lambda}{\beta\lambda+\frac{\lambda}{2^{6}\log^{2}n}}-\frac{\beta}{1+\beta}\right)\|P_{A}x\|_{2}^{2}.
\]

Using this, we have
\begin{align*}
f(x) & \defeq\frac{1}{t^{2}r^{2}}(P_{A}x)^{\top}(I-\frac{3}{4}P_{\widehat{A}})(P_{A}x)\\
 & \leq\left(1-\frac{3}{4}\frac{\beta\lambda}{\beta\lambda+\frac{\lambda}{2^{6}\log^{2}n}}+\frac{3}{4}\frac{\beta}{1+\beta}\right)\frac{\|x\|_{2}^{2}}{r^{2}}\\
 & \leq\left(1-\frac{3}{4}\frac{\beta}{\beta+\frac{1}{64\log^{2}n}}+\frac{3}{4}\frac{\beta}{1+\beta}\right)\cdot4\left(1-\frac{1}{\log n}\right)^{2}
\end{align*}
where we used $\|x\|_{2}\leq2(1-\frac{1}{\log n})r$ at the end. Putting
$\beta=\frac{1}{8\log n}$, we have
\begin{align*}
f(x) & \leq\left(1-\frac{3}{4}\cdot\frac{8\log n}{8\log n+1}+\frac{3}{4}\cdot\frac{1}{1+8\log n}\right)\cdot4\left(1-\frac{1}{\log n}\right)\\
 & =\left(\frac{7+8\log n}{8\log n+1}\right)\cdot\left(1-\frac{1}{\log n}\right)\\
 & \leq\left(1+\frac{1}{\log n}\right)\cdot\left(1-\frac{1}{\log n}\right)\\
 & \leq1.
\end{align*}
Hence from the definition of $\Omega_{1}$, we get
\begin{equation}
\frac{1}{t}P_{A}x\in\Omega_{1}.\label{eq:term_1_proof}
\end{equation}
Combining \eqref{eq:term_1_proof} and \eqref{eq:term_2_proof}, we
have
\[
x\in\text{Conv}(\Omega_{1}\cup\Omega_{2})\subset T'K.
\]
Therefore $T'K$ contains a ball of radius $2(1-\frac{1}{\log n})r$.
This completes the induction.
\end{proof}

With the bound on the inner radius and the $p$-norm of the covariance
matrix, we can apply Theorem \ref{thm:warm_start} and Lemma \ref{lem:KLS-anisotropic}
to bound the mixing time and the complexity of the algorithm $\mathsf{Isotropize}$.
\begin{thm}
\label{thm:well_rounded} The algorithm $\mathsf{Isotropize}$ applied
to a well-rounded input convex body $K$ satisfying $B(0,1)\subseteq K$
with $\E_{K}\norm x^{2}=\widetilde O(n)$, with high probability, finds a transformation
$T$ using $\widetilde{O}(n^{3}\psi^{2})$ oracle calls, s.t. $TK$
is $2$-isotropic.
\end{thm}

\begin{proof}
Theorem \ref{thm:sampling} shows that it takes $\widetilde{O}(n^{3})$
to get the first sample from a well-rounded body with Gaussian Cooling.
(Note that a uniform point from $K_{t}$ would be a very bad warm
start for $K_{2t}$.) This gives a warm start for all subsequent steps.

Let $A_{j}$ be the covariance matrix in the $j$-th iteration of
the algorithm. By Lemma \ref{lem:basic_property_isotropization} with
$\alpha=0$, we have that
\[
\norm{A_{j}}_{p}\le O(nr^{2/p}\log^{1+\frac{2}{p}} n).
\]
In the $j$-th iteration, the algorithm samples $k=c\cdot r_{j}^{2}\log^{7}n$
points. To bound the sampling cost, first note that for distribution
$q$ with KLS constant $\psi_{q}$, the complexity per sample is $\widetilde{O}(n^{2}\psi_{q}^{2}/r_{j}^{2})$
since by Lemma \ref{lem:innerball}, the algorithm maintains a ball
of radius $r_{j}$ inside the body. Suppose that the KLS constant
for isotropic distributions is $\psi$. We choose $p$ so that $\psi^{2}=O(n^{1/p})$.
By Lemma \ref{lem:KLS-anisotropic}, this implies
that the KLS constant even for non-isotropic distributions with covariance
$A$ satisfies $\psi_{q}^{2}=O(\norm A_{p}\log n)$. Lemma \ref{thm:warm_start}
shows that the total complexity of the $j$-th iteration is at most
\[
c\cdot r_{j}^{2}\cdot\widetilde{O}\left(\frac{n^{2}\cdot\norm{A_{j}}_{p}}{r_{j}^{2}}\right)\le\widetilde{O}\left(n^{3}r_{j}^{2/p}\right)\le\widetilde{O}\left(n^{3+\frac{1}{p}}\right).
\]
since we stop with $r^{2}=O(n)$. By the definition of $p$, the complexity
is $\widetilde{O}(n^{3}\psi^{2})$. For the cost of computing covariance
matrix and the mean, Lemma \ref{lem:covariance} shows that WHP $\widetilde O(n)$
samples suffice for a constant factor estimate of the covariance.
The total cost is
\[
n\cdot\widetilde{O}(n^{2}\cdot\norm A_{p}r^{-2})=r^{2}\cdot\widetilde{O}(n^{2}\cdot\norm A_{p}r^{-2})=\widetilde{O}\left(n^{3+\frac{1}{p}}\right)
\]
where we used that $n=\widetilde{O}(r^{2})$ at the end of the algorithm.
\end{proof}

\subsection{Outer loop: the general case}
For the general case, we first show that the next body is well-rounded
after each iteration and hence satisfies the condition of the algorithm
$\mathsf{Isotropize}$. 
\begin{lem}
\label{lem:round-3.5}Suppose that for a convex body $K$ and a transformation $T$, the set $T(K\cap B(0,t))-z$
is $2$-isotropic. Then $T\left(K\cap B\left(0,t(1+\frac{1}{\sqrt n})\right)\right)-z(1+\frac{1}{\sqrt{n}})$ is well-rounded.
\end{lem}
\begin{proof}
Let $K_{t}=K\cap B(0,t)$ and $K_{t(1+\frac{1}{\sqrt n})}=K\cap B\left(0,t(1+\frac{1}{\sqrt n})\right)$. Since $TK_t$ contains a ball of radius $1/2$, so does $TK_{t(1+\frac{1}{\sqrt n})}$. Next we prove that $\E_{TK_{t(1+\frac{1}{\sqrt n})}}\norm{x-z(1+\frac{1}{\sqrt{n}})}^2=O(n)$. 
By Lemma~\ref{lem:volume_outside_ball2}, we have $\vol(TK_t\setminus B(z,2\sqrt{n})\le \exp(-c\sqrt{n})\vol(TK_t)$ for some constant $c$. Since
\begin{align*}
TK_{t(1+\frac{1}{\sqrt n})}&\subseteq (1+\frac{1}{\sqrt n})TK_t\\
TK_{t(1+\frac{1}{\sqrt n})}\setminus (1+\frac{1}{\sqrt{n}})B(z,2\sqrt{n})&\subseteq (1+\frac{1}{\sqrt n})(TK_t\setminus B(z,2\sqrt{n})),
\end{align*}
we can compute the volume of $TK_{t(1+\frac{1}{\sqrt n})}$ that lies outside the ball $(1+\frac{1}{\sqrt{n}})B(z,2\sqrt{n})$ as follows:
\begin{align*}
\vol(TK_{t(1+\frac{1}{\sqrt n})}\setminus (1+\frac{1}{\sqrt{n}})B(z,2\sqrt{n}))
&\le (1+\frac{1}{\sqrt n})^n\vol(TK_t\setminus B(z,2\sqrt{n}))\\
&\le (1+\frac{1}{\sqrt n})^n \exp(-c\sqrt{n})\vol(TK_t)\\
&\le \exp(-(c-1)\sqrt{n})\vol(TK_{t(1+\frac{1}{\sqrt n})}).
\end{align*}
Moreover, since $TK_t-z$ is 2-isotropic, all of $TK_t$ lies in a ball of radius $2n$ and $z\in TK_t$. Then we have that $z(1+\frac{1}{\sqrt{n}})\in TK_{t(1+\frac{1}{\sqrt n})}$ and since $TK_{t(1+\frac{1}{\sqrt n})}\subseteq (1+\frac{1}{\sqrt n})TK_t$, $TK_{t(1+\frac{1}{\sqrt n})}$ lies in a ball of radius $(2n+2 \sqrt{n})$. Therefore, 
\[
\E_{TK_{t(1+\frac{1}{\sqrt n})}}\norm{x-z(1+\frac{1}{\sqrt{n}})}^2\le (2\sqrt{n}(1+\frac{1}{\sqrt{n}}))^2 + (2n+2\sqrt{n})^2\exp(-(c-1)\sqrt{n})\le 5n.
\]
\end{proof}

\begin{proof}[Proof of Theorem \ref{thm:rounding}.]
 The initial body $K_{r}$ is a ball, and at the end of the first
iteration, the body is $2$-isotropic. Let $K_{t}$ be the convex
body in some iteration. Assume that $K_{t}$ is $2$-isotropic. By
Lemma \ref{lem:round-3.5}, $K_{(1+1/\sqrt{n})t}$ is well-rounded. Since $(K_{t}-x)$
contains a unit ball, so does $(K_{(1+1/\sqrt{n})t}-x)$. Then by Theorem \ref{thm:well_rounded}
the inner loop to make $K_{(1+1/\sqrt{n})t}$ near-isotropic takes $\widetilde{O}(n^{3}\psi^{2})$
oracle queries. 
To bound the number of iterations of the outer loop, we consider the initial radius $r$ and it grows by a factor of $(1+\frac{1}{\sqrt{n}})$ in each iteration. When the algorithm stops, the number of iterations $\ell$ satisfies:
\[
\left(1+\frac{1}{\sqrt{n}}\right)^\ell r \ge R.
\]
Then we get $\ell = O(\sqrt{n}\log(R/r))$ and hence the theorem follows.
\end{proof}

\subsection{Further improvement}
The complexity of each inner iteration to make a well-rounded body isotropic is $\widetilde{O}(n^{3})$. Our current algorithm needs $\widetilde{O}(\sqrt{n})$ outer loop iterations, as we only scale the outer ball by a factor of $(1+(1/\sqrt{n}))$ in each iteration. The following conjecture, if true, would allow us to double the radius in each outer loop and hence reduce the number of iterations to $\widetilde{O}(1)$, thereby immediately improving the overall complexity of isotropic transformation to $O^*(n^3)$. To be precise, the following conjecture implies an almost cubic volume algorithm.
\begin{conjecture}
Suppose that for a convex body $K$ and an ellipsoid $E$, the set $K\cap E$
is $2$-isotropic. Then $K\cap 2E$ is well-rounded.
\end{conjecture}
This conjecture is true when $E$ is a ball, as shown in Lemma 3.4 in \cite{jia2021reducing}; however the case of a ball is not sufficient to obtain a cubic rounding algorithm for general bodies as claimed there. 

\section{Polytope Volume}

In this section, we consider the special case of convex polytopes.
We assume that a polytope $P=\left\{ x:Ax\le b\right\} $ is given
explicitly by $A\in R^{m\times n},b\in R^{m}$. A naive implementation
of the our general membership oracle algorithm would take $O(mn)$
arithmetic operations per oracle call, leading to a time bound of
$O^*(mn^{3.5})\cdot O(mn)=O^*(mn^{4.5})$, where the first term is the time complexity
from the oracle queries and the second term is from the additional
arithmetic operations per oracle query. This is now the number of
arithmetic operations. This already matches the current best time
complexity for polytopes using earlier volume algorithms and an amortization
trick introduced in \cite{mangoubi2019faster}. Here we give a significantly
faster implementation. The algorithm is based on a very simple
application of fast matrix multiplication. 

We can replace every ball walk sampling step with the following algorithm.
\begin{algorithm}[H]
\caption{Polytope Ball Walk}

\label{alg:poly_ball_walk} 

\begin{algorithmic}[1] 

\Procedure{$\mathsf{PolytopeBallWalk}$}{$A\in\R^{m\times n},b\in\R^{m},x_{0}\in\R^{n},\delta>0,k\in\mathbb{N}$} 

\State{\textbf{Assumption:} $Ax_{0}\le b$.}

\State{Sample points $z_{1},\dots,z_{k}$ from $B_{n}(0,\delta)$.} 

\State{$Z\leftarrow\begin{pmatrix}z_{1} & z_{2} & \cdots & z_{k}\end{pmatrix},Y\leftarrow\begin{pmatrix}Y_{1} & Y_{2} & \cdots & Y_{k}\end{pmatrix}\leftarrow AZ$.} 

\State{$x\leftarrow x_{0},y\leftarrow Ax_{0}$.} 

\For{$i\leftarrow1:k$} 

\If{$y+Y_{i}\le b$.}

\State{$x\leftarrow x+z_{i},y\leftarrow y+Y_{i}$.} 

\EndIf 

\EndFor 

\State{\textbf{Return} ($x$).} 

\EndProcedure

\end{algorithmic} 
\end{algorithm}

Next, we show that this algorithm will improve the running time of
the ball walk.
\begin{lem}
Given a polytope $\{x\in\R^{n}:Ax\le b\}$ where $A\in\R^{m\times n}$,
$k$ steps of the ball walk in the polytope can be implemented in
time $O(C(m,n,k))$ where $C(m,n,k)$ denotes the minimum number of
arithmetic operations needed to multiply an $m\times n$ matrix by
an $n\times k$ matrix.
\end{lem}

\begin{proof}
We can see that Algorithm \ref{alg:poly_ball_walk} does the same
computations as the ball walk algorithm with step size $\delta$ and
run for $k$ steps. The only difference is that we generate the random
vectors from $B(0,\delta)$ first and do a preprocessing step to avoid
multiplying $Az_{i}$ at each step. So the resulting point will be
the same as the ball walk algorithm and we reduce the running time
of each step. 

Let $Z,Y,x,y$ be the same as in Algorithm~\ref{alg:poly_ball_walk}.
Generating $k$ random vectors will cost $O(nk)$. For each step,
we compute $y+Y_{i}$ in $O(m)$ time and $x+z_{i}$ in $O(n)$ time.
So the total time taken is 
\[
O(nk)+C(m,n,k)+O((m+n)k)=O(C(m,n,k)).
\]
\end{proof}
We will apply this speedup to the polytope volume computation. From
Theorem~\ref{thm:well_rounded}, in the $j$-th iteration of Isotropiztion,
we run ball walk for $k=\widetilde{O}(n^{3.5})$ steps with the same
transformation matrix $T$, and hence we can use the above algorithm
to improve running time. 

For getting the first sample, we use the Gaussian Cooling algorithm
which in turn runs ball walk in phases with different ball radii and
target distributions. However, the total number of ball walk steps
in the Gaussian Cooling algorithm for a well-rounded body is $\widetilde{O}(n^{3})$
and we also know the maximum number of steps needed for each phase.
Therefore, we can construct the matrix $Z$ in Algorithm \ref{alg:poly_ball_walk}
with $k=\widetilde{O}(n^{3.5})$. The following lemma gives the current
best matrix bounds from a series of works \cite{gall2024faster,williams2024new} and it is computed using the tool \cite{Complexity}. For $k>0$, define the exponent of the rectangular
matrix multiplication as follows: 
\[
\omega(1,1,k)=\inf\{\tau\in\R\mid C(n,n,\lfloor n^{k}\rfloor)=O(n^{\tau})\}.
\]

\begin{fact}[\cite{gall2024faster,williams2024new}]
 $\omega(1,1,3.5)=4.6850095$.
\end{fact}
To prove the first part of Theorem \ref{thm:polytope}, use Algorithm
\ref{alg:poly_ball_walk}. For $\psi=\widetilde O(1)$, the total time is $\widetilde{O}(mn^{\omega(1,1,3.5)-1})=\widetilde{O}(mn^{3.6850095})$. 

\bigskip

{\bf Acknowledgements.} We thank Yunbum Kook and Kevin Tian for helpful comments. This work was supported in part by NSF Award CCF-2007443.

\bibliographystyle{plain}
\bibliography{main}

\begin{thebibliography}{10}

\bibitem{Adamczak2010}
R.~Adamczak, A.~Litvak, A.~Pajor, and N.~Tomczak-Jaegermann.
\newblock Quantitative estimates of the convergence of the empirical covariance matrix in log-concave ensembles.
\newblock {\em J. Amer. Math. Soc.}, 23:535--561, 2010.

\bibitem{AlmanW21}
Josh Alman and Virginia~Vassilevska Williams.
\newblock A refined laser method and faster matrix multiplication.
\newblock In D{\'{a}}niel Marx, editor, {\em Proceedings of the 2021 {ACM-SIAM} Symposium on Discrete Algorithms, {SODA} 2021, Virtual Conference, January 10 - 13, 2021}, pages 522--539. {SIAM}, 2021.

\bibitem{ApplegateK91}
D.~Applegate and R.~Kannan.
\newblock Sampling and integration of near log-concave functions.
\newblock In {\em STOC}, pages 156--163, 1991.

\bibitem{Complexity}
Jan van~den Brand.
\newblock Complexity term balancer.
\newblock \url{www.ocf.berkeley.edu/~vdbrand/complexity/}.
\newblock Tool to balance complexity terms depending on fast matrix multiplication.

\bibitem{chen2021almost}
Yuansi Chen.
\newblock An almost constant lower bound of the isoperimetric coefficient in the kls conjecture.
\newblock {\em Geometric and Functional Analysis}, pages 1--28, 2021.

\bibitem{CV2015}
B.~Cousins and S.~Vempala.
\newblock Bypassing {KLS}: {Gauss}ian cooling and an {$O^*(n^3)$} volume algorithm.
\newblock In {\em STOC}, pages 539--548, 2015.

\bibitem{CousinsV18}
Ben Cousins and Santosh~S. Vempala.
\newblock Gaussian cooling and {$O^*(n^3)$} algorithms for volume and gaussian volume.
\newblock {\em {SIAM} J. Comput.}, 47(3):1237--1273, 2018.

\bibitem{duan2023faster}
Ran Duan, Hongxun Wu, and Renfei Zhou.
\newblock Faster matrix multiplication via asymmetric hashing.
\newblock In {\em 2023 IEEE 64th Annual Symposium on Foundations of Computer Science (FOCS)}, pages 2129--2138. IEEE, 2023.

\bibitem{DyerF90}
M.~E. Dyer and A.~M. Frieze.
\newblock Computing the volume of a convex body: a case where randomness provably helps.
\newblock In {\em Proc. of AMS Symposium on Probabilistic Combinatorics and Its Applications}, pages 123--170, 1991.

\bibitem{DyerFK89}
M.~E. Dyer, A.~M. Frieze, and R.~Kannan.
\newblock A random polynomial time algorithm for approximating the volume of convex bodies.
\newblock In {\em STOC}, pages 375--381, 1989.

\bibitem{eldan2013thin}
Ronen Eldan.
\newblock Thin shell implies spectral gap up to polylog via a stochastic localization scheme.
\newblock {\em Geometric and Functional Analysis}, 23(2):532--569, 2013.

\bibitem{gall2024faster}
Fran{\c{c}}ois~Le Gall.
\newblock Faster rectangular matrix multiplication by combination loss analysis.
\newblock In {\em Proceedings of the 2024 Annual ACM-SIAM Symposium on Discrete Algorithms (SODA)}, pages 3765--3791. SIAM, 2024.

\bibitem{HCTFV2017}
Hulda~S Haraldsd\"ottir, Ben Cousins, Ines Thiele, Ronan~M.T Fleming, and Santosh Vempala.
\newblock {CHRR: coordinate hit-and-run with rounding for uniform sampling of constraint-based models}.
\newblock {\em Bioinformatics}, 33(11):1741--1743, 01 2017.

\bibitem{JLV2022}
Arun Jambulapati, Yin~Tat Lee, and Santosh~S. Vempala.
\newblock A slightly improved bound for the kls constant, {\em arxiv}, 2022.

\bibitem{jia2021reducing}
He~Jia, Aditi Laddha, Yin~Tat Lee, and Santosh Vempala.
\newblock Reducing isotropy and volume to kls: an o*(n 3 $\psi$ 2) volume algorithm.
\newblock In {\em Proceedings of the 53rd Annual ACM SIGACT Symposium on Theory of Computing}, pages 961--974, 2021.

\bibitem{jiang2019generalized}
Haotian Jiang, Yin~Tat Lee, and Santosh~S Vempala.
\newblock A generalized central limit conjecture for convex bodies.
\newblock {\em arXiv preprint arXiv:1909.13127}, 2019.

\bibitem{juditsky2008large}
Anatoli Juditsky and Arkadii~S Nemirovski.
\newblock Large deviations of vector-valued martingales in 2-smooth normed spaces.
\newblock {\em arXiv preprint arXiv:0809.0813}, 2008.

\bibitem{KLS95}
R.~Kannan, L.~Lov{\'a}sz, and M.~Simonovits.
\newblock Isoperimetric problems for convex bodies and a localization lemma.
\newblock {\em Discrete \& Computational Geometry}, 13:541--559, 1995.

\bibitem{KLS97}
R.~Kannan, L.~Lov\'{a}sz, and M.~Simonovits.
\newblock Random walks and an {$O^*(n^5)$} volume algorithm for convex bodies.
\newblock {\em Random Structures and Algorithms}, 11:1--50, 1997.

\bibitem{klartag2023logarithmic}
Bo'az Klartag.
\newblock Logarithmic bounds for isoperimetry and slices of convex sets.
\newblock {\em arXiv preprint arXiv:2303.14938}, 2023.

\bibitem{Klartag2022}
Bo'az Klartag and Joseph Lehec.
\newblock Bourgain's slicing problem and kls isoperimetry up to polylog, {\em arxiv}, 2022.

\bibitem{lee2018convergence}
Yin~Tat Lee and Santosh~S Vempala.
\newblock Convergence rate of riemannian hamiltonian monte carlo and faster polytope volume computation.
\newblock In {\em Proceedings of the 50th Annual ACM SIGACT Symposium on Theory of Computing}, pages 1115--1121. ACM, 2018.

\bibitem{lee2024eldan}
Yin~Tat Lee and Santosh~S Vempala.
\newblock Eldan's stochastic localization and the kls conjecture: Isoperimetry, concentration and mixing.
\newblock {\em Annals of Mathematics}, 199(3):1043--1092, 2024.

\bibitem{DBLP:journals/corr/LeeV16a}
Yin~Tat Lee and Santosh~Srinivas Vempala.
\newblock Eldan's stochastic localization and the {KLS} hyperplane conjecture: An improved lower bound for expansion.
\newblock {\em CoRR}, abs/1612.01507, 2016.

\bibitem{LeeV17KLS}
Yin~Tat Lee and Santosh~Srinivas Vempala.
\newblock Eldan's stochastic localization and the {KLS} hyperplane conjecture: An improved lower bound for expansion.
\newblock In {\em Proc. of IEEE FOCS}, 2017.

\bibitem{L90}
L.~Lov\'{a}sz.
\newblock How to compute the volume?
\newblock {\em Jber. d. Dt. Math.-Verein, Jubil\"aumstagung 1990}, pages 138--151, 1990.

\bibitem{LS90}
L.~Lov\'{a}sz and M.~Simonovits.
\newblock Mixing rate of {M}arkov chains, an isoperimetric inequality, and computing the volume.
\newblock In {\em ROCS}, pages 482--491, 1990.

\bibitem{LS93}
L.~Lov\'asz and M.~Simonovits.
\newblock Random walks in a convex body and an improved volume algorithm.
\newblock In {\em Random Structures and Alg.}, volume~4, pages 359--412, 1993.

\bibitem{LV2}
L.~Lov\'{a}sz and S.~Vempala.
\newblock Simulated annealing in convex bodies and an {$O^*(n^4)$} volume algorithm.
\newblock {\em J. Comput. Syst. Sci.}, 72(2):392--417, 2006.

\bibitem{LV07}
L.~Lov{\'a}sz and S.~Vempala.
\newblock The geometry of logconcave functions and sampling algorithms.
\newblock {\em Random Struct. Algorithms}, 30(3):307--358, 2007.

\bibitem{lovasz2006simulated}
L{\'a}szl{\'o} Lov{\'a}sz and Santosh Vempala.
\newblock Simulated annealing in convex bodies and an o*(n4) volume algorithm.
\newblock {\em Journal of Computer and System Sciences}, 72(2):392--417, 2006.

\bibitem{mangoubi2019faster}
Oren Mangoubi and Nisheeth~K Vishnoi.
\newblock Faster polytope rounding, sampling, and volume computation via a sub-linear ball walk.
\newblock In {\em 2019 IEEE 60th Annual Symposium on Foundations of Computer Science (FOCS)}, pages 1338--1357. IEEE, 2019.

\bibitem{Paouris2006}
G.~Paouris.
\newblock Concentration of mass on convex bodies.
\newblock {\em Geometric and Functional Analysis}, 16:1021--1049, 2006.

\bibitem{srivastava2013covariance}
Nikhil Srivastava and Roman Vershynin.
\newblock Covariance estimation for distributions with 2+eps moments.
\newblock {\em The Annals of Probability}, 41(5):3081--3111, 2013.

\bibitem{tropp2012user}
Joel~A Tropp.
\newblock User-friendly tail bounds for sums of random matrices.
\newblock {\em Foundations of computational mathematics}, 12(4):389--434, 2012.

\bibitem{williams2024new}
Virginia~Vassilevska Williams, Yinzhan Xu, Zixuan Xu, and Renfei Zhou.
\newblock New bounds for matrix multiplication: from alpha to omega.
\newblock In {\em Proceedings of the 2024 Annual ACM-SIAM Symposium on Discrete Algorithms (SODA)}, pages 3792--3835. SIAM, 2024.

\end{thebibliography}

\appendix

\section{Empirical Covariance Matrix with Sublinear Sample Complexity\label{sec:matrix_chernoff}}

To bound the error of the empirical covariance matrix $A$, we use
the following matrix Chernoff bound.
\begin{lem}[{Matrix Bernstein \cite[Theorem 6.1]{tropp2012user}}]
\label{lem:chernoff}Consider a finite sequence $\{X_{i}\}$ of independent,
random, self-adjoint matrices of dimension $n$. Assume that $\E X_{i}=0$
and $\|X_{i}\|_{\op}\leq R$ almost surely. Then, for all $t\geq0$,
we have
\[
\P\left(\|\sum_{i}X_{i}\|_{\op}\geq t\right)\leq2n\exp\left(\frac{-t^{2}/2}{\sigma^{2}+Rt/3}\right)
\]
where $\sigma^{2}\defeq\|\sum_{i}\E(X_{i}^{2})\|_{\op}.$ 
\end{lem}

\begin{lem}
\label{lem:chernoff_app}Let $p$ be a logconcave density in $\Rn$
with covariance $A$. Let $\widehat{A}\leftarrow\frac{1}{k}\sum_{i=1}^{k}(x_{i}-\widehat{x})(x_{i}-\widehat{x})^{\top}$
where $\widehat{x}=\frac{1}{k}\sum_{i=1}^{k}x_{i+k}$ and $x_{i}$
are independent samples from $p$. With probability $1-1/n^{O(1)}$,
for any $0\leq\epsilon\leq1$, we have
\[
(1-\epsilon)A-O\left(\frac{\log^{3}n\tr A}{\epsilon k}\right)\cdot I\preceq\widehat{A}\preceq(1+\epsilon)A+O\left(\frac{\log^{3}n\cdot\tr A}{\epsilon k}\right)\cdot I.
\]
\end{lem}

\begin{rem*}
By a more careful tail analysis, one can get a bound of $O(\frac{\log n\cdot\tr A}{\epsilon k})\cdot I$
and we speculate that the tight bound for the additive term might
be $O(\frac{\tr A}{\epsilon k})\cdot I$.
\end{rem*}
\begin{proof}
Let $\lambda\geq0$ to some constant to be determined. By shifting
the distribution, we can assume $p$ has mean $0$. Let $\widetilde{p}$
be the distribution given by $p$ restricted to the ball 
\[
\{\|x\|_{(\lambda A+I)^{-1}}\leq3s\cdot\sqrt{\tr(\lambda A+I)^{-\frac{1}{2}}A(\lambda A+I)^{-\frac{1}{2}}}\}
\]
for some $s=\Theta(\log n)$. Using the fact that $p$ has mean $0$
and the fact that $(\lambda A+I)^{-\frac{1}{2}}A(\lambda A+I)^{-\frac{1}{2}}\preceq A$,
Lemma \ref{lem:volume_outside_ball} shows that
\[
\P_{x\sim p}(\|x\|_{(\lambda A+I)^{-1}}\leq3s\cdot\sqrt{\tr A})\geq1-\frac{1}{n^{\Theta(1)}}.
\]

Let $A_{i}$ be the random matrices $(x_{i}-\widehat{x})(x_{i}-\widehat{x})^{\top}$
with $x_{i}$ sampled from $p$ and $\widetilde{A}_{i}$ be the random
matrices $(\widetilde{x}_{i}-\widehat{x})(\widetilde{x}_{i}-\widehat{x})^{\top}$
with $\widetilde{x}_{i}$ sampled from $\widetilde{p}$. Note that
when $k=\Omega(n)$, we have that WHP $\frac{1}{2}A\preceq\widehat{A}\preceq2A$
\ref{lem:covariance}. Hence, we can assume $k=O(n)$. We couple two
matrices together such that $\widetilde{A}_{i}=A_{i}$ for $i=1,2,\cdots k$
with probability $1-ke^{-s}=1-\frac{1}{n^{\Theta(1)}}$. Let $X_{i}=\frac{1}{k}(\lambda A+I)^{-\frac{1}{2}}(A_{i}-\E A_{i})(\lambda A+I)^{-\frac{1}{2}}$
and $\widetilde{X}_{i}=\frac{1}{k}(\lambda A+I)^{-\frac{1}{2}}(\widetilde{A}_{i}-\E\widetilde{A}_{i})(\lambda A+I)^{-\frac{1}{2}}$
for some $\lambda$ to be chosen where the expectation $\E$ is conditional
on $\hat{x}$. Note that
\begin{align}
\P(\norm{(\lambda A+I)^{-\frac{1}{2}}(\widehat{A}-\E\widehat{A})(\lambda A+I)^{-\frac{1}{2}}}_{\op}\geq t) & =\P(\|\sum_{i}X_{i}\|_{\op}\geq t)\nonumber \\
 & \leq\P(\|\sum_{i}\widetilde{X}_{i}\|_{\op}\geq t)+\frac{1}{n^{\Theta(1)}}.\label{eq:A_EA_truncation}
\end{align}
Hence, it suffices to study $\widetilde{X}_{i}$.

We will apply Lemma \ref{lem:chernoff}. For the sup norm bound, note
that
\begin{align*}
\|\widetilde{X}_{i}\|_{\op} & \leq\frac{1}{k}\left(\|(\lambda A+I)^{-\frac{1}{2}}\widetilde{A}_{i}(\lambda A+I)^{-\frac{1}{2}}\|_{\op}+\|(\lambda A+I)^{-\frac{1}{2}}\E\widetilde{A}_{i}(\lambda A+I)^{-\frac{1}{2}}\|_{\op}\right)\\
 & \leq\frac{1}{k}\left(\|\widetilde{x}_{i}-\widehat{x}\|_{(\lambda A+I)^{-1}}^{2}+\E\|\widetilde{x}_{i}-\widehat{x}\|_{(\lambda A+I)^{-1}}^{2}\right)\\
 & \leq\frac{1}{k}\left(2\|\widetilde{x}_{i}\|_{(\lambda A+I)^{-1}}^{2}+2\E\|\widetilde{x}_{i}\|_{(\lambda A+I)^{-1}}^{2}+4\|\widehat{x}\|_{(\lambda A+I)^{-1}}^{2}\right)\\
 & \leq\frac{4}{k}\left(9s^{2}\tr A+\|\widehat{x}\|_{(\lambda A+I)^{-1}}^{2}\right)\leq R
\end{align*}
where the second inequality follows from the fact that non-zero eigenvalues
of $Y^{T}Y$ and $YY^{T}$ are the same for any rectangular matrix
$Y$, and
\begin{equation}
R\defeq\frac{36}{k}\left(s^{2}\tr A+\|\widehat{x}\|_{(\lambda A+I)^{-1}}^{2}\right)\label{eq:A_R}
\end{equation}
For the variance bound, note that $(\lambda A+I)^{-\frac{1}{2}}\widetilde{A}_{i}(\lambda A+I)^{-\frac{1}{2}}\preceq kR$
and hence (using again that $(a-b)^{2}\le2(a^{2}+b^{2}))$,
\begin{align*}
\E\widetilde{X}_{i}^{2} & \preceq\E\frac{2}{k^{2}}(((\lambda A+I)^{-\frac{1}{2}}\widetilde{A}_{i}(\lambda A+I)^{-\frac{1}{2}})^{2}+(\E(\lambda A+I)^{-\frac{1}{2}}\widetilde{A}_{i}(\lambda A+I)^{-\frac{1}{2}})^{2})\\
 & \preceq\frac{2kR}{k^{2}}\E((\lambda A+I)^{-\frac{1}{2}}\widetilde{A}_{i}(\lambda A+I)^{-\frac{1}{2}}+\E((\lambda A+I)^{-\frac{1}{2}}\widetilde{A}_{i}(\lambda A+I)^{-\frac{1}{2}})\\
 & =\frac{4R}{k}\E(\lambda A+I)^{-\frac{1}{2}}\widetilde{A}_{i}(\lambda A+I)^{-\frac{1}{2}}\\
 & \preceq\frac{8R}{k}\E(\lambda A+I)^{-\frac{1}{2}}A_{i}(\lambda A+I)^{-\frac{1}{2}}\\
 & =\frac{8R}{k}(\frac{A}{\lambda A+I}+(\lambda A+I)^{-\frac{1}{2}}\widehat{x}\widehat{x}^{\top}(\lambda A+I)^{-\frac{1}{2}})
\end{align*}
where we used that $\frac{d\widetilde{p}}{dp}\leq2$ in the last inequality.
Hence, we have 
\[
\sigma^{2}\defeq\|\sum_{i}\E\widetilde{X}_{i}^{2}\|_{\op}\leq8R(\|\frac{A}{\lambda A+I}\|_{\op}+\|\widehat{x}\|_{(\lambda A+I)^{-1}}^{2}).
\]
Apply Lemma \ref{lem:chernoff}, with probability $1-\frac{1}{n^{\Theta(1)}}$,
we have

\begin{align*}
\|\sum_{i}\widetilde{X}_{i}\|_{\op} & \lesssim\sigma\sqrt{\log n}+R\log n\\
 & \lesssim\sqrt{(\|\frac{A}{\lambda A+I}\|_{\op}+\|\widehat{x}\|_{(\lambda A+I)^{-1}}^{2})}\sqrt{R\log n}+R\log n
\end{align*}

Using the value of $R$ from equation \ref{eq:A_R}, for any $c\geq1$,
we get 
\begin{align*}
\|\sum_{i}\widetilde{X}_{i}\|_{\op} & \lesssim\frac{1}{c}\|\frac{A}{\lambda A+I}\|_{\op}+\frac{1}{c}\|\widehat{x}\|_{(\lambda A+I)^{-1}}^{2}+\frac{c}{k}(\log^{2}n\cdot\tr A+\|\widehat{x}\|_{(\lambda A+I)^{-1}}^{2})\log n\\
 & \lesssim\frac{1}{c\lambda}+(\frac{1}{c}+\frac{c\log n}{k})\cdot\|\widehat{x}\|_{(\lambda A+I)^{-1}}^{2}+\frac{c\log^{3}n}{k}\cdot\tr A
\end{align*}
Using this and equation \eqref{eq:A_EA_truncation} , we have
\[
\norm{(\lambda A+I)^{-\frac{1}{2}}(\widehat{A}-\E\widehat{A})(\lambda A+I)^{-\frac{1}{2}}}_{\op}\lesssim\frac{1}{c\lambda}+(\frac{1}{c}+\frac{c\log n}{k})\cdot\|\widehat{x}\|_{(\lambda A+I)^{-1}}^{2}+\frac{c\log^{3}n}{k}\cdot\tr A.
\]
Finally, we note that $\widehat{x}$ follows a logconcave distribution
with mean $0$ and covariance matrix $\frac{1}{k}A$. By Lemma \ref{lem:volume_outside_ball},
we have that 
\[
\|\widehat{x}\|_{(\lambda A+I)^{-1}}\lesssim\log n\sqrt{\frac{1}{k}\tr A}
\]
with probability $1-1/n^{O(1)}$. This gives 
\begin{align*}
\norm{(\lambda A+I)^{-\frac{1}{2}}(\widehat{A}-\E\widehat{A})(\lambda A+I)^{-\frac{1}{2}}}_{\op} & \lesssim\frac{1}{c\lambda}+(\frac{\log^{2}n}{ck}+\frac{c\log^{3}n}{k})\cdot\tr A\\
 & \lesssim\frac{1}{c\lambda}+\frac{c\log^{3}n}{k}\cdot\tr A
\end{align*}
Hence, we have
\[
|\widehat{A}-\E\widehat{A}|\lesssim\frac{1}{c}A+\frac{1}{c\lambda}I+\frac{c\log^{3}n}{k}\cdot\tr A\cdot(\lambda A+I).
\]
Taking $c=\Theta(\epsilon^{-1})$, $\lambda=\frac{k}{c^{2}\log^{3}n\tr A}$,
we have
\[
|\widehat{A}-\E\widehat{A}|\preceq\epsilon A+O(\frac{\log^{3}n\tr A}{\epsilon k})\cdot I.
\]
The result follows using $\E\widehat{A}=A$.
\end{proof}

\section{Proof of the Anisotropic KLS Bound\label{sec:KLS}}

Consider the following stochastic localization process.
\begin{defn}
\label{def:A}For a logconcave density $p$, we define the following
stochastic differential equation:
\begin{equation}
c_{0}=0,\quad dc_{t}=dW_{t}+\mu_{t}dt,\label{eq:dBt}
\end{equation}
where the probability density $p_{t}$, the mean $\mu_{t}$ and the
covariance $A_{t}$ are defined by
\[
p_{t}(x)=\frac{e^{c_{t}^{\top}x-\frac{t}{2}\norm x_{2}^{2}}p(x)}{\int_{\Rn}e^{c_{t}^{\top}y-\frac{t}{2}\norm y_{2}^{2}}p(y)dy},\quad\mu_{t}=\E_{x\sim p_{t}}x,\quad A_{t}=\E_{x\sim p_{t}}(x-\mu_{t})(x-\mu_{t})^{\top}.
\]

The following lemma shows that one can upper bound the expansion $\psi_{p}$
by upper bounding $\|A_{t}\|_{\op}$:
\end{defn}

\begin{lem}[{\cite[Lemma 31 in ArXiv ver 3]{DBLP:journals/corr/LeeV16a}}]
\label{lem:boundAgivesKLS}Given a logconcave density $p$, let $A_{t}$
be as in Definition \ref{def:A} using initial density $p$. Suppose
there is a $T>0$ such that
\[
\P\left(\int_{0}^{T}\norm{A_{s}}_{\op}ds\leq\frac{1}{64}\right)\geq\frac{3}{4}
\]
Then, we have $\psi_{p}=O\left(T^{-1/2}\right).$
\end{lem}

To bound $\|A_{t}\|_{\op}$, we need a basic stochastic calculus rule
about $A_{t}$.
\begin{lem}[{\cite[Lemma 27 in arXiv ver 3]{DBLP:journals/corr/LeeV16a}}]
\label{lem:def-pt}The covariance $A_{t}$ satisfies 
\[
dA_{t}=\int_{\Rn}(x-\mu_{t})(x-\mu_{t})^{\top}\left((x-\mu_{t})^{\top}dW_{t}\right)p_{t}(x)dx-A_{t}^{2}dt.
\]
\end{lem}

We bound $\|A_{t}\|_{\op}$ using the potential $\Phi_{t}\defeq\tr A_{t}^{q}$.
We will use fractional $q$ and there is no simple formula of $d\Phi_{t}$.
The upper bound on $d\Phi_{t}$ is known (see \cite{eldan2013thin}).
For completeness, we give an alternative proof here. Our proof relies
on the following lemma about the smoothness of the trace function.
\begin{lem}[\cite{juditsky2008large}, Proposition 3.1]
\label{lem:twice_matrix_func_upper}Let $f$ be a twice differentiable
function on $(\alpha,\beta)$ such that for some $\theta_{\pm},\mu_{\pm}\in\R$,
for all $\alpha<a<b<\beta$, we have
\[
\theta_{-}\frac{f''(a)+f''(b)}{2}+\mu_{-}\leq\frac{f'(b)-f'(a)}{b-a}\leq\theta_{+}\frac{f''(a)+f''(b)}{2}+\mu_{+}.
\]
Then, for any matrix $X$ with eigenvalues in $(\alpha,\beta)$, we
have
\[
\theta_{-}\tr(f''(X)H^{2})+\mu_{-}\tr H^{2}\leq\left.\frac{\partial^{2}\tr f(X)}{\partial X^{2}}\right|_{H,H}\leq\theta_{+}\tr(f''(X)H^{2})+\mu_{+}\tr H^{2}.
\]
\end{lem}

Now, we can use this to upper bound the derivative of $\Phi_{t}$.
\begin{lem}
For any $q>1$, we have that
\begin{align*}
d\Phi_{t}\leq & q\E_{x\sim p_{t}}(x-\mu_{t})^{\top}A_{t}^{q-1}(x-\mu_{t})(x-\mu_{t})^{\top}dW_{t}\\
 & +q(q-1)\cdot\E_{x,y\sim p_{t}}((x-\mu_{t})^{\top}(y-\mu_{t}))^{2}(x-\mu_{t})^{\top}A_{t}^{q-2}(y-\mu_{t})dt.
\end{align*}
\end{lem}

\begin{proof}
By Lemma \ref{lem:def-pt}, we have that
\begin{align*}
dA_{t} & =\int_{\Rn}(x-\mu_{t})(x-\mu_{t})^{\top}\left((x-\mu_{t})^{\top}dW_{t}\right)p_{t}(x)dx-A_{t}^{2}dt\\
 & =\sum_{i}Z_{i}\cdot dW_{t,i}-A_{t}^{2}dt
\end{align*}
with $Z_{i}\defeq\E_{x\sim p_{t}}(x-\mu_{t})(x-\mu_{t})^{\top}(x-\mu_{t})_{i}$.
By Itô's formula, we have that
\[
d\Phi_{t}=\left.\frac{\partial\tr A_{t}^{q}}{\partial A_{t}}\right|_{dA_{t}}+\frac{1}{2}\sum_{i}\left.\frac{\partial^{2}\tr A_{t}^{q}}{\partial A_{t}^{2}}\right|_{Z_{i},Z_{i}}dt.
\]
For the first-order term, we have
\begin{align*}
\left.\frac{\partial\tr A_{t}^{q}}{\partial A_{t}}\right|_{dA_{t}} & =q\cdot\tr A_{t}^{q-1}dA_{t}\\
 & =q\E_{x\sim p_{t}}(x-\mu_{t})^{\top}A_{t}^{q-1}(x-\mu_{t})(x-\mu_{t})^{\top}dW_{t}-q\tr A_{t}^{q+1}.
\end{align*}
For the second-order term, we use Lemma \ref{lem:twice_matrix_func_upper}
with $f(x)=x^{q}$. It is easy to see that $\frac{f'(b)-f'(a)}{b-a}\leq f''(a)+f''(b).$
Hence, we can use Lemma \ref{lem:twice_matrix_func_upper} with $\theta_{+}=2$
and $\mu_{+}=0$ to get
\[
\left.\frac{\partial^{2}\tr A_{t}^{q}}{\partial A_{t}^{2}}\right|_{Z_{i},Z_{i}}\leq2q(q-1)\tr(A_{t}^{q-2}Z_{i}^{2}).
\]
Combining the first and second-order terms, we have
\begin{align*}
d\Phi_{t}\leq & q\E_{x\sim p_{t}}(x-\mu_{t})^{\top}A_{t}^{q-1}(x-\mu_{t})(x-\mu_{t})^{\top}dW_{t}\\
 & +q(q-1)\cdot\sum_{i}\tr(A_{t}^{q-2}Z_{i}^{2})dt\\
= & q\E_{x\sim p_{t}}(x-\mu_{t})^{\top}A_{t}^{q-1}(x-\mu_{t})(x-\mu_{t})^{\top}dW_{t}\\
 & +q(q-1)\cdot\E_{x,y\sim p_{t}}((x-\mu_{t})^{\top}(y-\mu_{t}))^{2}(x-\mu_{t})^{\top}A_{t}^{q-2}(y-\mu_{t})dt.
\end{align*}
\end{proof}
To analyze the stochastic inequality for $d\Phi_{t}$, we introduce
a 3-Tensor.
\begin{defn}[3-Tensor]
For an isotropic logconcave distribution $p$ in $\Rn$ and symmetric
matrices $A,B$ and $C$, define
\begin{align*}
T(A,B,C) & =\sup_{\text{isotropic log-concave }p}\E_{x,y\sim p}(x^{\top}Ay)(x^{\top}By)(x^{\top}Cy).
\end{align*}
\end{defn}

Using the definition above, we can simplify the upper bound of $\Phi_{t}$
as follows:
\begin{equation}
d\Phi_{t}\leq q\E_{x\sim p_{t}}(x-\mu_{t})^{\top}A_{t}^{q-1}(x-\mu_{t})(x-\mu_{t})^{\top}dW_{t}+q(q-1)\cdot T(A_{t}^{q-1},A_{t},A_{t}).\label{eq:dPhi_bound}
\end{equation}
To further bounding $d\Phi_{t}$, we need following inequalities about
logconcave distributions:
\begin{lem}[{\cite[Lemma 32 in arXiv ver 3]{DBLP:journals/corr/LeeV16a}}]
\label{lem:inq1}Let $p$ be a logconcave density with mean $\mu$
and covariance $A$. For any positive semi-definite matrix $C$, we
have that
\[
\|\E_{x\sim p}(x-\mu)(x-\mu)^{\top}C(x-\mu)\|_{2}\lesssim\|A\|_{\op}^{1/2}\tr(A^{1/2}CA^{1/2}).
\]
\end{lem}

\begin{lem}[{\cite[Lemma 41]{jiang2019generalized}}]
\label{lem:inq2}For any $0\leq\alpha\leq1$, $A\succeq0$, and $C\succeq0$,
we have $T(A^{\alpha},A^{1-\alpha},C)\leq T(A,I,C)$.
\end{lem}

\begin{lem}[{\cite[Lemma 40]{jiang2019generalized}}]
\label{lem:inq3}Suppose that $\psi_{n}\leq\alpha n^{\beta}$ for
all $n$ with some fixed $0\leq\beta\leq\frac{1}{2}$ and $\alpha\geq1$.
For any two symmetric matrices $A$ and $B$, we have
\[
T(A,B,I)\lesssim\alpha^{2}\log n\cdot\|A\|_{1}\cdot\|B\|_{1/(2\beta)}.
\]
\end{lem}

Using the lemmas above, we have the following:
\begin{lem}
\label{lem:Phi_upper}Suppose that $\psi_{n}=O(n^{\beta})$ for all
$n$ with some fixed $0\leq\beta\leq\frac{1}{2}$. For any $q=\frac{1}{2\beta}$,
we have $d\Phi_{t}\le\delta_{t}dt+v_{t}^{\top}dW_{t}$ with $\delta_{t}\lesssim\frac{\log n}{\beta^{2}}\cdot\Phi_{t}^{1+2\beta}$
and $\|v_{t}\|_{2}\lesssim\frac{1}{\beta}\cdot\Phi_{t}^{1+\beta}$.
\end{lem}

\begin{proof}
From (\ref{eq:dPhi_bound}), we have
\begin{align*}
\delta_{t} & \leq q(q-1)\cdot T(A_{t}^{q-1},A_{t},A_{t})\\
 & \leq q(q-1)\cdot T(A_{t}^{q},A_{t},I)\\
 & \lesssim q(q-1)\cdot\log n\cdot\tr A_{t}^{q}\cdot(\tr A_{t}^{1/(2\beta)})^{2\beta}\\
 & \lesssim\frac{\log n}{\beta^{2}}\cdot\Phi_{t}^{1+2\beta}
\end{align*}
where we used Lemma \ref{lem:inq2} in the second inequality and Lemma
\ref{lem:inq3} at the end. From (\ref{eq:dPhi_bound}) and Lemma
\ref{lem:inq1}, we have
\begin{align*}
\|v_{t}\|_{2} & \defeq q\|\E_{x\sim p_{t}}(x-\mu_{t})^{\top}A_{t}^{q-1}(x-\mu_{t})(x-\mu_{t})\|_{2}\\
 & \lesssim q\|A_{t}\|_{\op}^{1/2}\cdot\tr A_{t}\lesssim\frac{1}{\beta}\cdot\Phi_{t}^{1+\beta}.
\end{align*}
\end{proof}
Finally, we use the next lemma to bound stochastic inequality and apply it in the proof of the next theorem.
\begin{lem}[{\cite[Lemma 35]{jiang2019generalized}}]
\label{lem:stochastic_inq}Let $\Phi_{t}$ be a stochastic process
such that $\Phi_{0}\leq\frac{U}{2}$ and $d\Phi_{t}=\delta_{t}dt+v_{t}^{\top}dW_{t}$.
Let $T>0$ be some fixed time, $U>0$ be some target upper bound,
and $f$ and $g$ be some auxiliary functions such that for all $0\leq t\leq T$
\begin{enumerate}
\item $\delta_{t}\leq f(\Phi_{t})$ and $\|v_{t}\|_{2}\leq g(\Phi_{t})$,
\item Both $f(\cdot)$ and $g(\cdot)$ are non-negative non-decreasing functions,
\item $f(U)\cdot T\leq\frac{U}{8}$ and $g(U)\cdot\sqrt{T}\leq\frac{U}{8}$.
\end{enumerate}
Then, we have that $\P\left[\max_{t\in[0,T]}\Phi_{t}\geq U\right]\leq0.01$.
\end{lem}

\begin{thm}
\label{thm:KLS_ani}Suppose that $\psi_{n}=O(n^{\beta})$ for all
$n$ with some fixed $0\leq\beta\leq\frac{1}{2}$. Then, for any logconcave
distribution $p$ with covariance matrix $A$, we have that
\[
\psi_{p}\lesssim\frac{\sqrt{\log n}}{\beta}\|A\|_{\frac{1}{2\beta}}^{1/2}.
\]
\end{thm}

\begin{proof}
Consider the stochastic process starts with $p_{0}=p$. Let $\Phi_{t}=\tr A_{t}^{1/2\beta}$.
Lemma \ref{lem:Phi_upper} shows that $d\Phi_{t}=\delta_{t}dt+v_{t}^{\top}dW_{t}$
with $\delta_{t}\leq\frac{\log n}{\beta^{2}}\cdot\Phi_{t}^{1+2\beta}$
and $\|v_{t}\|_{2}\leq\frac{1}{\beta}\cdot\Phi_{t}^{1+\beta}$. Let
$U=2\tr A^{1/2\beta}$, $f(\Phi)=c\frac{\log n}{\beta^{2}}\Phi^{1+2\beta}$,
and $g(\Phi)=c\frac{1}{\beta}\cdot\Phi^{1+\beta}$ for some large
enough constant $c$. Take $T=c'\frac{\beta^{2}}{U^{2\beta}\log n}$
with some small enough constant $c'$. Note that $f(U)\cdot T\leq\frac{U}{8}$
and $g(U)\cdot\sqrt{T}\leq\frac{U}{8}$. This verifies the conditions
in Lemma \ref{lem:stochastic_inq} and hence this shows that
\[
\P\left[\max_{t\in[0,T]}\Phi_{t}\geq U\right]\leq0.01.
\]
Hence, we have $\P\left[\max_{t\in[0,T]}\|A_{t}\|_{\op}\geq U^{2\beta}\right]\leq0.01$
and 
\[
\P\left(\int_{0}^{T}\norm{A_{s}}_{\op}ds\leq\frac{1}{64}\right)\geq\frac{3}{4}.
\]
Using this, Lemma \ref{lem:boundAgivesKLS} shows that
\[
\psi_{p}\lesssim T^{-1/2}\lesssim\frac{1}{\beta}U^{\beta}\sqrt{\log n}\lesssim\frac{\sqrt{\log n}}{\beta}\|A\|_{\frac{1}{2\beta}}^{1/2}.
\]
\end{proof}

\end{document}